\documentclass[11pt]{article}

\usepackage{natbib}
\usepackage{graphicx}%
\usepackage{multirow}%
\usepackage{amsmath,amssymb,amsfonts}%
\usepackage{amsthm}%
\usepackage{mathrsfs}%
\usepackage[title]{appendix}%
\usepackage{xcolor}%
\usepackage{textcomp}%
\usepackage{manyfoot}%
\usepackage{booktabs}%

\usepackage{listings}%
\usepackage{a4wide}

\usepackage{bbm}
\usepackage{bm}
\usepackage{subcaption}

\usepackage{authblk}   

\newcommand{\E}{\mathbb E}
\newcommand{\e}{\mathrm e}

\newcommand{\D}{\mathrm{d}}

\theoremstyle{definition}
\newtheorem{theorem}{Theorem}
\newtheorem{lemma}[theorem]{Lemma}

\newtheorem{proposition}[theorem]{Proposition}
\newtheorem{remark}[theorem]{Remark}

\newtheorem{definition}{Definition}

\begin{document}

\title{Forecasting duration in high-frequency financial data using a self-exciting flexible residual point process}

\author{Kyungsub Lee$^{1}$}

\affil{$^{1}$Department of Statistics, Yeungnam University, Gyeongsan, Republic of Korea}

\date{} %

\maketitle

\begin{abstract}
This paper presents a method for forecasting limit order book durations using a self-exciting flexible residual point process.
High-frequency events in modern exchanges exhibit heavy-tailed interarrival times, posing a significant challenge for accurate prediction. 
The proposed approach incorporates the empirical distributional features of interarrival times while preserving the self-exciting and decay structure. 
This work also examines the stochastic stability of the process,
which can be interpreted as a general state-space Markov chain.
Under suitable conditions, the process is irreducible, aperiodic, positive Harris recurrent, and has a stationary distribution. 
An empirical study demonstrates that the model achieves strong predictive performance compared with several alternative approaches when forecasting durations in ultra-high-frequency trading data.
\end{abstract}
\vspace{1em}

\section{Introduction}~\label{Sec:intro}

This paper has two primary objectives. 
The first is to predict the durations observed in the limit order book (LOB) using a self-exciting and exponentially decaying point process with a flexible residual component. 
An LOB refers to the dataset that records all types of orders, cancellations, and transactions submitted to an exchange. 
In modern stock exchanges, such as the New York Stock Exchange (NYSE) and the National Association of Securities Dealers Automated Quotations (NASDAQ), millions of orders are recorded on a nanosecond scale each day. 
The intervals between these ultra-high-frequency events are called durations, and various types of durations can be defined.

This study focuses on the durations associated with changes in the mid-price, which provide critical insights into market microstructure dynamics by capturing the timing and intensity of order flow.
Variations in trade arrival rates, as demonstrated by \cite{easley2008time}, reflect information flow and price discovery, 
and econometric frameworks such as autoregression and point process models reviewed by \cite{bauwens2009modelling} capture temporal clustering in these durations.

Price-change durations are also closely related to volatility, and several studies have explored this connection.
For example, \cite{hong2023volatility} developed duration-based volatility estimators, and \cite{fei2023forecasting} and \cite{slim2023forecasting} demonstrated improved forecasting performance when incorporating price durations.
Complementary approaches, including the stochastic duration model by \cite{pelletier2024stochastic} and the duration-dependent Markov-switching framework by \cite{turatti2025combining}, argued that price-change durations are important for modeling and forecasting volatility, liquidity, and order flow dynamics.

Predicting the future durations of mid-price changes at a high-frequency level is highly challenging, since the distribution is not mathematically tractable (e.g., as in the exponential distribution). 
Instead, the distribution of these changes exhibits heavy tails, with a combination of extremely short and long intervals.
Such extreme behavior makes accurate prediction particularly difficult.

Despite these difficulties, predicting the duration based on stochastic processes can be structured into two steps. 
The first step is to identify and characterize the dependence structure of interarrival times. 
Such dependence can be described from a clustering perspective.
The occurrence of short interarrivals tends to increase the likelihood of subsequent short interarrivals, and the same holds for long intervals.
The second step is to select an appropriate distribution for stochastic modeling. 
These procedures result in reduced-form models constructed based on the empirical characteristics of the data.

According to \cite{hautsch2011econometrics}, statistical modeling approaches to this problem can be broadly classified into four categories: intensity, hazard, duration, and count data models. 
Although these approaches appear different regarding their mathematical formulations, they share the common objective of modeling the probability of an event occurring at a given time.

A representative example of an intensity model that describes the number of events occurring within a time interval is the Hawkes process~\citep{Hawkes1, Hawkes2}. 
Although many variations of the Hawkes process exist, its fundamental feature is that each event increases the intensity by a fixed amount, after which the intensity gradually decays over time. 
The rate of decay depends on the random interarrival time until the next event. 
The Hawkes process has an appealing mathematical structure,
which can be represented using the immigration–descendant framework
and has been widely applied in finance, particularly for modeling high-frequency market dynamics \citep{embrechts2011multivariate,bacry2015hawkes,lee2023modeling,chen2024modeling,jain2024limit}.
Stationarity can be ensured by requiring that the integral of the kernel to be less than one. 
However, if an alternative kernel specification~\citep{hardiman2013critical, bacry2016estimation, kanazawa2021ubiquitous} replaces the exponential decay kernel,
the process loses its finite-state Markov property, 
leading to substantial difficulties for simulation and estimation.

Duration models such as the autoregressive conditional duration (ACD) model introduced by \cite{Engle1998}, have been widely applied and extended to numerous variants \citep{bauwens2000logarithmic,fernandes2006family,meitz2006evaluating}.
Unlike the Hawkes process, where generally a constant jump size is added to the event frequency, 
duration models determine the increments of the underlying stochastic process randomly at each event time.
The process is specified as a weighted average of past durations and represents the expected value of the next interarrival time. 
Owing to this simple structure, the exponential distribution originally employed in the ACD model as well as various alternative distributions can be incorporated \citep{hautsch2003assessing,bauwens2004comparison,de2006regime}.

According to \cite{lee2025self}, intensity models can also incorporate a wide range of distributions while preserving the fundamental self-exciting and decay terms. 
This approach naturally embeds the empirical distributional features of interarrival times in the intensity framework,  enhancing the accuracy of future predictions. 
When applied to LOB durations, the idea is to maintain the self-exciting and decaying structure of the Hawkes process and Markov property by introducing flexible residuals.
These residuals are defined as independent and identically distributed (i.i.d.) sequences of heavy-tailed random variables transformed via a suitable inverse function.
This work also discusses how duration models and flexible residual models can be viewed as belonging to a single family of models.

The second focus of this paper is the stochastic stability of the flexible residual point process. 
Although \cite{lee2025self} originally proposed this model, only its definition and basic properties were introduced, 
and its fundamental probabilistic characteristics remain largely unexplored. 
While the process is presented as an intensity-based point process, 
it can also be viewed as a Markov chain on a general state-space.

Under certain conditions, an appropriate Foster–Lyapunov function can be constructed for a general state-space Markov chain, ensuring that the chain is positive Harris recurrent.
Moreover, the chain is irreducible and aperiodic, and it possesses a stationary distribution, making it suitable for rigorous probabilistic analysis.

The remainder of this paper is organized as follows. 
Section~\ref{Sec:model} introduces the model, discusses its basic properties, and compares it with alternative approaches. 
Section~\ref{Sect:SEFRPP} examines the stochastic stability from the perspective of Markov chains. 
Section~\ref{Sect:empirical} presents an empirical study based on high-frequency LOB data and compares the duration forecasting performance of several models. 
Finally, Section~\ref{Sect:concl} concludes the paper.

\section{Modeling framework}\label{Sec:model}

\subsection{General setup}

This section outlines the mathematical framework for modeling event arrivals as a point process. 
While a point process can be formally defined using a measure-theoretic approach, we focus here on the sequence of event times and the underlying dynamics that drive them. 
For a rigorous treatment of the measure-theoretic foundations of point processes,
we refer the interested reader to \cite{daley2007introduction} and \cite{bremaud2020point}.

In the context of this study, we consider a point process defined on the positive real line $\mathbb R_+$, representing a sequence of arrival times. 
The process can be characterized by its arrival times, $ 0 < T_1 < T_2 < \cdots$, and interarrival times $\tau_i = T_i - T_{i-1}$, where the distribution of each  depends on the history of the process.

To model these dynamics, we adopt a framework where the interarrival times are governed by an underlying state process that evolves over time and depends only on past information.
This approach allows us to capture the self-exciting or autoregressive properties often observed in high-frequency data. 
The following definition introduces a flexible class of models for interarrival times that exhibit the Markov property.

\begin{definition}\label{Def:general}
	Let $\{ \varepsilon_n \}_{n \in \mathbb{N}}$ be a sequence of i.i.d. random variables with an absolutely continuous distribution on $(0,\infty)$. 
	Suppose there exists an underlying state process $\{ X_n \}$ adapted to the natural filtration
	$
	\mathcal{F}_{n-1} = \sigma\bigl(X_0, \varepsilon_1, \ldots, \varepsilon_{n-1}\bigr).
	$
	The function \( \Phi(t, x) \) is assumed to be strictly increasing in \( t \) for each fixed \( x \) and defined for \( t \geq 0 \).
	The interarrival time $\tau_n$ and $\varepsilon_n$ have a relationship such that
	\begin{equation}
	\varepsilon_n = \Phi(\tau_n, X_{n-1}) \quad \text{i.e.,} \quad \tau_n = \Phi^{-1}(\varepsilon_n, X_{n-1})  \label{Eq:tau}
	\end{equation}
	where the inverse is taken with respect to \( t \), holding \( x \) fixed.
	The dynamics evolve according to an update function $\Psi(t, x)$ such that
	\begin{equation}
		X_n = \Psi(\tau_n, X_{n-1}). \label{Eq:X}
	\end{equation}
	Then, a locally finite point process $N$ on $\mathbb R_{+}$ is defined by
	$$ N = \sum_{n \in \mathbb N} \delta_{T_n}, \quad \text{where} \quad T_n = \sum_{i=1}^n \tau_i.$$
\end{definition}

Since the process $\{X_n\}$ is adapted to the natural filtration,
the resulting point process is evolutionary (non-anticipative)
in the standard sense.
The conditional expectations of durations and
intensities are taken with respect to past information only, as is
customary in the point process and duration-modeling literature.
In particular, the initial state $X_0$ is $\mathcal{F}_0$-measurable
and does not depend on future innovations.

It is important to note that the framework defined above belongs to the class of
observation-driven models, as categorized by \cite{Cox1981}.
In this specification, the latent process $\{X_n\}$ evolves as a deterministic
function of its past values and past realized durations, with stochastic innovations
entering the system solely through the sequence $\{\varepsilon_n\}$ that directly
drives the interarrival times.

This structure encompasses a wide range of popular duration models used in forecasting applications.
At the same time, it deliberately rules out additional, independent sources of
stochasticity in the latent state dynamics.
In this respect, the proposed framework differs from parameter-driven models, in which
the latent process is itself subject to separate stochastic innovations.

Examples of such parameter-driven specifications include the stochastic conditional
intensity model \citep{Bauwens2006}, the stochastic multiplicative error model
\citep{Hautsch2008}, and various Markov-switching duration and intensity models
\citep{Chen2013, Li2021}.

We also note that the baseline framework in Definition~\ref{Def:general} is formulated in event time and does not explicitly incorporate calendar-time effects.
In practice, however, diurnal and seasonal effects in high-frequency financial durations are well documented \citep{hautsch2011econometrics}.

Calendar-time effects can be incorporated into the proposed framework in several ways.
First, durations may be multiplicatively decomposed by pre-filtering the raw durations $\tau_n$ with a deterministic seasonal component $s(T_n)$, for example estimated using cubic splines, yielding de-seasonalized durations $\tilde{\tau}_n = \tau_n / s(T_n)$.
The proposed framework can then be applied to the adjusted durations, as is standard in the literature.

Second, calendar-time dependence may be introduced directly by allowing the transformation and state-update functions to depend on event time, leading to specifications of the form
$\Phi(\tau_n, X_{n-1}, T_{n-1})$ and $\Psi(\tau_n, X_{n-1}, T_{n-1})$.
Alternatively, selected parameters, such as the baseline component, may be modeled as smooth functions of clock time.

Finally, as in our empirical analysis, intraday time variation is addressed using a rolling-window estimation strategy.
By re-estimating model parameters over fixed intervals (e.g., every 5{,}000 observations), the baseline intensity and dynamic responses adapt locally to evolving market conditions throughout the trading day.
This provides a flexible, data-driven approach without imposing restrictive parametric assumptions on seasonality.

\subsection{Specific models}

The exact behavior of the process in Definition~\ref{Def:general} depends on the specific choice of $\Phi$ and $\Psi$;
hence, a more detailed analysis is deferred to later sections, focusing on special cases of primary interest.
Let $F_\varepsilon$, $f_\varepsilon$ and  $h_\varepsilon$ be the cumulative distribution function (CDF), probability density function (PDF) and hazard function, respectively, of $\varepsilon$.
The conditional intensity for $t = t' - T_{n-1}$ with $t' \in (T_{n-1}, T_{n}]$ is represented by
\begin{align}
	\lambda_{n-1}(t) &= \lim_{\Delta t \rightarrow 0} \frac{\mathbb E[ N(t' + \Delta t) - N(t')|  \mathcal F_{n-1}]}{\Delta t} \nonumber \\
	&= \lim_{\Delta t \rightarrow 0} \frac{\mathbb P[ t < \tau_n < t + \Delta t|   \mathcal F_{n-1}]}{\Delta t} \frac{1}{\mathbb P(\tau_n > t |  \mathcal F_{n-1}) } \nonumber \\
	&=  \lim_{\Delta t \rightarrow 0} \frac{F_{\varepsilon} (\Phi(t + \Delta t, X_{n-1})) - F_{\varepsilon} (\Phi(t, X_{n-1}))}{\Delta t} \frac{1}{\mathbb P(\tau_n > t |  \mathcal F_{n-1})}  \nonumber \\
	&= \frac{\partial F_{\varepsilon}(\Phi(t, X_{n-1}))}{\partial t} \frac{1}{\mathbb P(\tau_n > t |  \mathcal F_{n-1})} 
	= \frac{ f_{\varepsilon}(\Phi(t, X_{n-1})) }{1 - F_{\varepsilon}(\Phi(t, X_{n-1}))} \frac{\partial \Phi(t, X_{n-1})}{\partial t}  \nonumber\\
	&= h_\varepsilon(\Phi(t, X_{n-1})) \frac{\partial \Phi(t, X_{n-1})}{\partial t}. \label{Eq:intensity}
\end{align}
We write $\lambda_{n-1}(t)$ as a function of $t = t' - T_{n-1}$, representing the elapsed time since the $(n-1)$th event,
which is a slight abuse of the notation adopted for convenience.
The CDF and PDF of $\tau = \tau_n$ with $x = X_{n-1}$ are respectively represented by
\begin{equation}
F_\tau(t; x) = F_\varepsilon(\Phi(t, x)), \quad f_\tau(t;x) = f_\varepsilon(\Phi(t,x)) \frac{\partial \Phi(t,x)}{\partial t}. \label{Eq:F_tau}
\end{equation}

In the point process literature, this quantity is known as the backward recurrence time, formally defined as $x(t') = t' - T_{\tilde{N}(t')}$, where $\tilde{N}(t') := \sum_{n \geq 0} \mathbf{1}\{T_n < t'\}$ is the left-continuous version of the counting process \citep{hautsch2011econometrics}.

Throughout this paper, all models are expressed as functions of the elapsed time $t = t' - T_{n-1}$ through the integrated intensity $\Phi(t, X_{n-1})$, given the state $X_{n-1}$ updated
at each event time.
This includes the Hawkes process, which is conventionally formulated as a function of calendar time $t'$ by accumulating contributions from all past events.
In the special case of an exponential kernel, however, the Hawkes process admits a Markovian representation in which the full history is summarized by the intensity at the most recent event time,
allowing it to be expressed within our elapsed-time framework.

\begin{definition}[Renewal process]
	If $\Phi(t, x) = t $, then $\tau_n = \varepsilon_n$ and the resulting process is called a renewal process \citep{feller1950introduction}.
\end{definition}

\begin{definition}[Autoregressive conditional duration (ACD) model]\label{Def:ACD}
	Let $\{ X_n \}$ evolve according to the linear update rule, corresponding to Eq.~\eqref{Eq:X},
	\[
	X_n = \Psi(\tau_n, X_{n-1}) = b_0 + a \tau_n + b_1 X_{n-1}, 
	\]
	where $a, b_0, b_1 > 0$ and $a + b_1 < 1$ to ensure stationarity.
	If
	\begin{equation}
	\Phi(t, x) = \frac{t}{x}, \quad\quad\text{i.e.,}\quad\quad \Phi^{-1}(y, x) = xy\label{Eq:phi_ACD}
	\end{equation}
	then the interarrival time is given by
	\[
	\tau_n = \Phi^{-1}( \varepsilon_n, X_{n-1}) = X_{n-1} \varepsilon_n.
	\]
	Then the resulting process is known as the ACD(1,1) model introduced in~\cite{Engle1998}.
	The ACD(1,1) model is a special case  where the latent process ${X_n}$ is the conditional expected duration.
\end{definition}

\begin{definition}[Log-ACD(1,1) model]\label{Def:logACD}
	Let $\{ X_n \}$ evolve according to
	\[
	X_n = \Psi(\tau_n, X_{n-1}) = \exp \left( b_0 + a \log(\tau_n) + b_1 \log(X_{n-1}) \right)
	\]
	where $a,b_1 \ge 0$, $b_0 \in \mathbb{R}$, and $a+b_1<1$.
	with $\Phi(t, x) = t/x$ as in Definition~\ref{Def:ACD}.
	Then the resulting process is known as the log-ACD(1,1) model introduced by \cite{bauwens2000logarithmic}.
\end{definition}

In the original ACD model, durations are typically written as $\tau_n = X_n \varepsilon_n$.
Here, we adopt the convention $\tau_n = X_{n-1}\varepsilon_n$ to maintain consistency across models.
This change is purely notational and emphasizes that the conditional expected duration $X_{n-1}$ is determined at time $T_{n-1}$.

With unit exponential innovations, the ACD model implies a constant conditional intensity on $(T_{n-1},T_n)$,
\[
\lambda_{n-1}(t)=\frac{1}{X_{n-1}}.
\]
When $\varepsilon_n$ is non-exponential, the intensity varies with elapsed time according to
\begin{equation}
	\lambda_{n-1}(t)
	= h_\varepsilon\!\left(\frac{t}{X_{n-1}}\right)\frac{1}{X_{n-1}}.
	\label{Eq:intensity_ACD}
\end{equation}
(For a detailed discussion of the relationship between ACD models and autoregressive conditional intensity models, 
see \cite{hautsch2011econometrics}.)

\begin{definition}[Log-ACI(1,1) model]\label{Def:logACI}
	Let $\{\varepsilon_n\}$ be an i.i.d.\ sequence of standard exponential 
	random variables.
	Let the latent process $\{\Lambda_n\}$ evolve according to
	\[
	\Lambda_n 
	= \exp \left( b_0
	+ a\bigl(\Lambda_{n-1}\tau_n - 1\bigr)
	+ b_1 \log \Lambda_{n-1} \right),
	\]
	where $a, b_0 \in \mathbb{R}$ and $|b_1| < 1$.
	Assume that the transformation function $\Phi(t,x)$ is given by
	\begin{equation}
		\Phi(t,x) = x t,
		\quad\quad\text{i.e.,}\quad\quad
		\Phi^{-1}(y,x) = \frac{y}{x}.
		\label{Eq:phi_ACI}
	\end{equation}
	Then the interarrival time satisfies
	\begin{equation}
		\tau_n = \Phi^{-1}(\varepsilon_n, \Lambda_{n-1})
		= \frac{\varepsilon_n}{\Lambda_{n-1}}. \label{Eq:tau_ACI}
	\end{equation}
	The resulting point process is known as the log-ACI(1,1) model introduced 
	in~\cite{russell1999econometric}, in which $\{\Lambda_n\}$ represents 
	the conditional intensity process.
\end{definition}

\begin{definition}[Self-exciting exponentially decaying flexible residual point process]\label{Def:FRPP}
	The latent process $\{ \Lambda_n \}$ evolves according to
	\begin{equation}
		\Psi(t, x) = \mu + (x - \mu + \alpha)\e^{-\beta t}, \qquad \mu, \alpha, \beta > 0.
		\label{Eq:psi_FRPP}
	\end{equation}
	Thus,
	\[
	\Lambda_n = \Psi(\tau_n, \Lambda_{n-1})
	= \mu + (\Lambda_{n-1} - \mu + \alpha)\e^{-\beta \tau_n}.
	\]
	We define $\Phi(t, x)$ as follows:
	\begin{equation}
		\Phi(t, x)
		= \int_0^t \Psi(s, x)\,\mathrm{d}s
		= \mu t + (x - \mu + \alpha)\frac{1 - \e^{-\beta t}}{\beta}.
		\label{Eq:phi_FRPP}
	\end{equation}
	The interarrival times are generated by
	\[
	\tau_n = \Phi^{-1}(\varepsilon_n, \Lambda_{n-1}),
	\]
	where the inverse is taken with respect to $t$.
	The resulting process is a self-exciting, exponentially decaying point process with flexible residuals introduced in~\cite{lee2025self}.
\end{definition}

\begin{remark}\label{Remark:stability}
	To ensure that the latent process $\{\Lambda_n\}$ is well defined
	and admits a stationary regime, we impose the condition
	\[
	\alpha < \beta \mu_\varepsilon,
	\qquad
	\mu_\varepsilon := \E[\varepsilon] < \infty.
	\]
	See Section~\ref{Sect:SEFRPP} for further discussion.
\end{remark}

Under Definition~\ref{Def:FRPP}, the conditional intensity between events is
\begin{equation}
	\lambda_{n-1}(t)
	= h_{\varepsilon}(\Phi(t, \Lambda_{n-1})) \, \Psi(t, \Lambda_{n-1}),
	\label{Eq:intensity_FRPP}
\end{equation}
where $\Psi$ and $\Phi$ are given in
Eqs.~\eqref{Eq:psi_FRPP}--\eqref{Eq:phi_FRPP}.

\begin{definition}[Exponential Hawkes process]
	If $\varepsilon_n$ follows the unit exponential distribution,
	the process in Definition~\ref{Def:FRPP} reduces to the
	exponential Hawkes process \citep{Hawkes1,Hawkes2},
	since $h_{\varepsilon}(\cdot) = 1$ in
	Eq.~\eqref{Eq:intensity_FRPP}.
\end{definition}

Definitions~\ref{Def:ACD} and~\ref{Def:FRPP} highlight the differences between ACD-type models and self-exciting models.

In self-exciting models, the occurrence of an event induces a deterministic upward
shift in the latent process, which directly affects the conditional intensity.
If the subsequent arrival occurs after a long delay, the latent process decays
gradually toward its baseline level.
In contrast, when arrivals occur in rapid succession, the decay between events is
limited, allowing the latent process—and hence the intensity—to build up over time.
Repeated clusters of short interarrival times therefore lead to substantial
self-excitation.
In ACD-type models, by contrast, the latent process evolves only at event times,
and its update is stochastic.
A long observed interarrival time leads to an increase in the latent duration,
which in turn implies a lower conditional intensity due to their inverse relationship.
Conversely, a short interarrival time reduces the latent duration and results in
a higher conditional intensity.

In summary, self-exciting models feature deterministic jumps in the latent process
at event times, with randomness entering through the length-dependent decay between
arrivals.
ACD-type models, on the other hand, maintain a constant latent process between
events, while randomness enters through the event-driven updates determined by
the preceding interarrival time.

More generally, Definition~\ref{Def:general} provides a unified framework for
constructing a broad class of point processes through suitable choices of the
functions $\Psi$ and $\Phi$.
This flexibility enables the modeler to tailor the dynamics to different data
features or application settings.
Among the many possible specifications, this work focuses on the recently proposed
process in Definition~\ref{Def:FRPP}, which is examined in detail in
Sections~\ref{Sect:SEFRPP} and~\ref{Sect:empirical}.

\subsection{Residual process}

Integrating the conditional intensity $\lambda_{n-1}(t)$ over the interarrival
time $\tau_n$ yields an exponential residual,\footnote{This residual arises from
the random time change theorem \citep{Meyer1971,BrownNair1988} and is distinct from
the model innovation $\varepsilon_n$.}
which coincides with the transformed time in intensity-based models.
Specifically,
\begin{align}
	\int_0^{\tau_n} \lambda_{n-1}(t) \, \mathrm{d}t
	&= \int_0^{\tau_n} 
	\frac{f_{\varepsilon}(\Phi(t, X_{n-1}))}{1 - F_{\varepsilon}(\Phi(t, X_{n-1}))}  
	\frac{\partial \Phi}{\partial t}(t, X_{n-1})  \, \D t \\
	&= \int_0^{\varepsilon_n} \frac{f_{\varepsilon}(u)}{1 - F_{\varepsilon}(u)} \, \mathrm{d}u \\
	&= - \log(1 - F_{\varepsilon}(\varepsilon_n)) = - \log(1 - F_{\varepsilon}(\Phi(\tau_n, X_{n-1}))). 	\label{Eq:eres}
\end{align}
Since $F_\varepsilon(\varepsilon_n)$ is uniformly distributed on $(0,1)$,
the residual in~\eqref{Eq:eres} follows a standard exponential distribution.
Equivalently,
\[
1-\exp\!\left(-\int_0^{\tau_n}\lambda_{n-1}(t)\,\mathrm dt\right)
= F_\varepsilon(\varepsilon_n).
\]

The conditional expectation of the interarrival time given $\mathcal F_{n-1}$ is
\begin{equation}
	\E[\tau_n \mid \mathcal{F}_{n-1}] = \int_0^\infty s \, \lambda_{n-1}(s) \, S_{n-1}(s) \D s, \label{Eq:slambdaS}
\end{equation}
where \( S_{n-1}(s) \) denotes the conditional survival function,
\begin{equation}
	S_{n-1}(s) = \exp\left(-\int_0^s \lambda_{n-1}(u) \D u\right) = 1 - F_{\varepsilon}(\Phi(s, X_{n-1})).
\end{equation}
Using \( \lambda_{n-1}(s) S_{n-1}(s)
= f_\varepsilon(\Phi(s,X_{n-1})) \partial \Phi/\partial s \),
we obtain
\begin{equation}
	\E[\tau_n \mid \mathcal{F}_{n-1}] =  \int_0^{\infty} s f_{\varepsilon}(\Phi(s, X_{n-1}))
	\frac{\partial \Phi}{\partial s} (s, X_{n-1}) \D s.
\end{equation}
When the innovation distribution has a finite mean
(and hence \(\E[\tau_n \mid \mathcal F_{n-1}] < \infty\)),
from Eq.~\eqref{Eq:slambdaS} we also have
\begin{align}
	\E[\tau_n \mid \mathcal{F}_{n-1}] & = \int_0^\infty s \lambda_{n-1}(s) S_{n-1}(s) \D s \\
	&= \left[-s S_{n-1}(s)\right]_0^\infty +  \int_0^\infty S_{n-1}(s) \D s = \int_0^\infty S_{n-1}(s) \D s \\
	&= \int_0^\infty
	\left[1 - F_{\varepsilon} \left(\Phi(s, X_{n-1})\right)\right] \D s. \label{Eq:cond_tau_SE2}
\end{align}

In the ACD model, where $\Phi(s, X_{n-1}) = s / X_{n-1}$,
\begin{equation}
	\E[\tau_n \mid \mathcal{F}_{n-1}] =  \int_0^{\infty} s f_{\varepsilon}\left(\frac{s}{X_{n-1}}\right) 
	\frac{\mathrm{d}}{\mathrm{d}s} \left( \frac{s}{X_{n-1}} \right) \D s = X_{n-1} \mu_{\varepsilon}. \label{Eq:cond_tau_ACD}
\end{equation}
For the self-exciting and exponentially decaying flexible residual point process, we have
\begin{equation}
	\E[\tau_n \mid \mathcal{F}_{n-1}] = \int_0^\infty s f_{\varepsilon}(\Phi(s, \Lambda_{n-1})) \Psi(s, \Lambda_{n-1}) \D s, \label{Eq:cond_tau_SE} 
\end{equation}
where \( \Psi(s, \Lambda_{n-1}) = \lambda_{n-1}(s) \) and \( f_{\varepsilon} \) denotes the density of the innovation variable.

If $\varepsilon$ follows the standard exponential distribution (i.e., $f_{\varepsilon}(x) = \e^{-x}$), this equation simplifies to
\begin{align*}
	\E[\tau_n \mid \mathcal{F}_{n-1}] &= \int_0^\infty s \e^{- \Phi(s, \Lambda_{n-1})}  \Psi(s, \Lambda_{n-1}) \D s \\
	&= \int_0^\infty  \e^{- \Phi(s, \Lambda_{n-1})}   \D s.
\end{align*}

The conditional expectation $\E[\tau_n \mid \mathcal F_{n-1}]$ plays a central role
in our empirical analysis to forecast the next interarrival time.
Except for special cases such as the ACD model, the integral expressions for
$\E[\tau_n \mid \mathcal F_{n-1}]$ do not admit closed-form solutions.
In the empirical implementation, these expectations are therefore evaluated
numerically using one-dimensional numerical integration, conditional on the
estimated latent state.

\subsection{Log-likelihood function}

The log-likelihood function based on the observed interarrival times $\{ \tau_n \}_{n=1}^{N}$ is given by
\[
\log \mathcal{L}_N(\bm\theta) = \sum_{n=1}^{N} \ell_n(\bm\theta),
\]
where $\ell_n(\bm\theta)$ denotes the contribution to the log-likelihood from the $n$th observation, and is expressed as
\begin{equation}
	\ell_n(\bm\theta) = \log f_{\varepsilon}(\Phi(\tau_n, X_{n-1})) + \log \left( \frac{\partial \Phi}{\partial t}(\tau_n, X_{n-1}) \right). \label{Eq:ln1-general}
\end{equation}
In the self-exciting point process with flexible residuals, where $\Psi(t, x) = \frac{\partial \Phi}{\partial t}(t, x)$, the expression simplifies to
\begin{equation}
	\ell_n(\bm\theta) = \log f_{\varepsilon}(\Phi(\tau_n, X_{n-1})) + \log \Psi(\tau_n, X_{n-1}). \label{Eq:ln1-flexible}
\end{equation}
Then the maximum likelihood estimator is defined by
\begin{equation}
	\hat{\bm\theta} = \arg\max_{\bm\theta} \log \mathcal{L}_N(\bm\theta)
	= \arg\max_{\bm\theta} \sum_{n=1}^{N} \left\{ \log f_{\varepsilon}(\Phi(\tau_n, X_{n-1})) + \log \left( \frac{\partial \Phi}{\partial t}(\tau_n, X_{n-1}) \right) \right\}. \label{Eq:MLE}
\end{equation}
This likelihood formulation follows the standard construction for point processes
based on conditional intensities, see \citet{Karr1991}.

\section{Self-exciting point process with flexible residual}\label{Sect:SEFRPP}

\subsection{Linear representation}
This section focuses on the properties of the self-exciting and exponentially decaying flexible residual point process
defined in~Definition~\ref{Def:FRPP}.

By construction, the process $\Lambda_n$ evolves according to the exponential decay, moderating the self-exciting effect. 
Specifically,
\begin{equation}
	\Lambda_n = \mu + (\Lambda_{n-1} - \mu + \alpha)\e^{-\beta \tau_n}. \label{Eq:lambda}
\end{equation}
Using the transformation formula in Eq.~\eqref{Eq:phi_FRPP} yields
\[
\varepsilon_n = \Phi(\tau_n, \Lambda_{n-1}) = \mu \tau_n + (\Lambda_{n-1} - \mu + \alpha)\frac{1 - \e^{-\beta \tau_n}}{\beta},
\]
from which it follows that
\[
(\Lambda_{n-1} - \mu + \alpha)\e^{-\beta \tau_n} = \Lambda_{n-1} - \mu + \alpha + \beta \mu \tau_n - \beta \varepsilon_n.
\]
Substituting this equation into Eq.~\eqref{Eq:lambda} yields the following linear representation:
\begin{equation}
	\Lambda_n = \Lambda_{n-1} + \alpha - \beta \varepsilon_n + \beta \mu \tau_n. \label{Eq:lambda_linear}
\end{equation}
Recursively applying Eq.~\eqref{Eq:lambda_linear} leads to the series expansion
\begin{equation}
	\Lambda_n = \Lambda_{0} + n\alpha - \beta \sum_{i=1}^n\varepsilon_i + \beta \mu \sum_{i=1}^{n}\tau_i. \label{Eq:lambda_n_series}
\end{equation}
Taking expectations in Eq.~\eqref{Eq:lambda_n_series} yields
\[
\mathbb{E}[\Lambda_n]
= \Lambda_0 + n(\alpha - \beta \mu_\varepsilon)
+ \beta \mu \sum_{i=1}^n \mathbb{E}[\tau_i].
\]
Since $\mathbb{E}[\tau_i] > 0$, a necessary condition for $\mathbb{E}[\Lambda_n]$ to remain bounded is
\begin{equation}
	\alpha \le \beta \mu_\varepsilon,
	\label{Eq:stability}
\end{equation}
which motivates the stability condition imposed in Remark~\ref{Remark:stability}.

\subsection{Markov Chain}

This section examines the process $\{\Lambda_n\}$ as a general state-space Markov chain, 
which can also be considered a scalar nonlinear state-space model.
Definitions and standard background follow \cite{meyn2009markov}.

The state space of the Markov chain is given by $\mathcal{X} = [\mu, \infty)$. 
Even if the chain starts at $\mu$, it cannot return to this value, since the update rule includes a strictly positive excitation term $\alpha$ and an exponential decay factor smaller than one.

The transition dynamics are characterized by a transition kernel \( P(x, B) \) for any $x \in \mathcal{X}$ and any Borel set \( B \subseteq \mathcal{X} \), defined as
\begin{equation}
P(x, B) = \mathbb{P}\left( \mu + (x - \mu + \alpha) \e^{-\beta \tau} \in B \right), \label{Eq:P}
\end{equation}
where \( \tau = \Phi^{-1}(\varepsilon, x) \), and \( \varepsilon \sim F_\varepsilon \) represents a positive random variable with support on $(0, \infty)$ and density function \( f_\varepsilon > 0 \).

Under the stability condition, the fundamental asymptotic and statistical properties of the intensity process are established in the following theorem.

\begin{theorem}[Ergodicity and Stationarity]\label{Thm:Main}
	Suppose the stability condition $\alpha < \beta \mu_\varepsilon$ holds. Then, the process $\{ \Lambda_n \}$ is a Lebesgue-irreducible, aperiodic, and positive Harris recurrent Markov chain on the state space $\mathcal{X} = [\mu, \infty)$. Specifically:
	\begin{enumerate}
		\item There exists a unique stationary distribution $\pi$. For any initial distribution, the law of $\Lambda_n$ converges to $\pi$ in total variation as $n \to \infty$.
		\item The process is ergodic, ensuring that sample averages of the durations $\{ \tau_n \}$ converge to their theoretical counterparts.
		\item  Under stationarity, the expected inter-arrival time is uniquely determined by the model parameters:
		\begin{equation}
			\mathbb{E}[\tau_n] = \frac{\beta \mu_\varepsilon - \alpha}{\beta \mu}. \label{Eq:Exp_Tau}
		\end{equation}
	\end{enumerate}
\end{theorem}

\begin{proof}
	See Appendix~\ref{Append:math}.
\end{proof}

The properties established in Theorem~\ref{Thm:Main} provide the theoretical foundation
for forecasting with the proposed model.
In particular, ergodicity guarantees the consistency of likelihood-based estimation,
while the stability condition $\alpha < \beta \mu_\varepsilon$ prevents explosive
intensity dynamics and ensures well-defined long-run duration forecasts.

\subsection{Autocorrelation structure}

\begin{figure}[t]
	\centering
	\begin{subfigure}[t]{0.45\textwidth}
		\includegraphics[width=\textwidth]{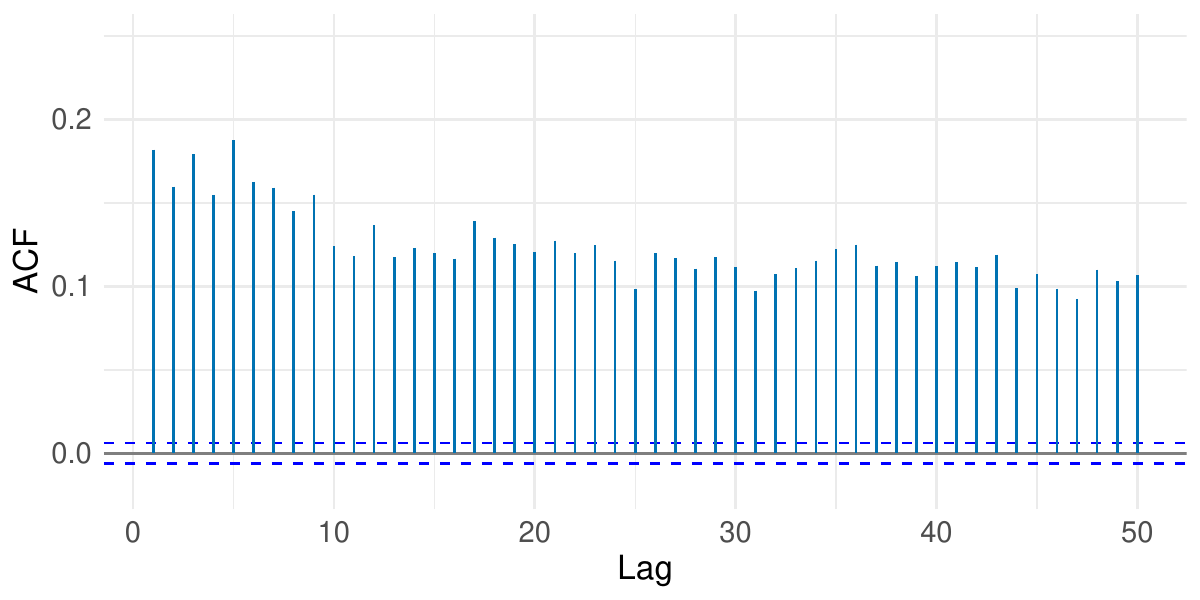}
		\caption{$\alpha = 0.095$ and $\beta = 0.1$}
	\end{subfigure}
	\quad
	\begin{subfigure}[t]{0.45\textwidth}
		\includegraphics[width=\textwidth]{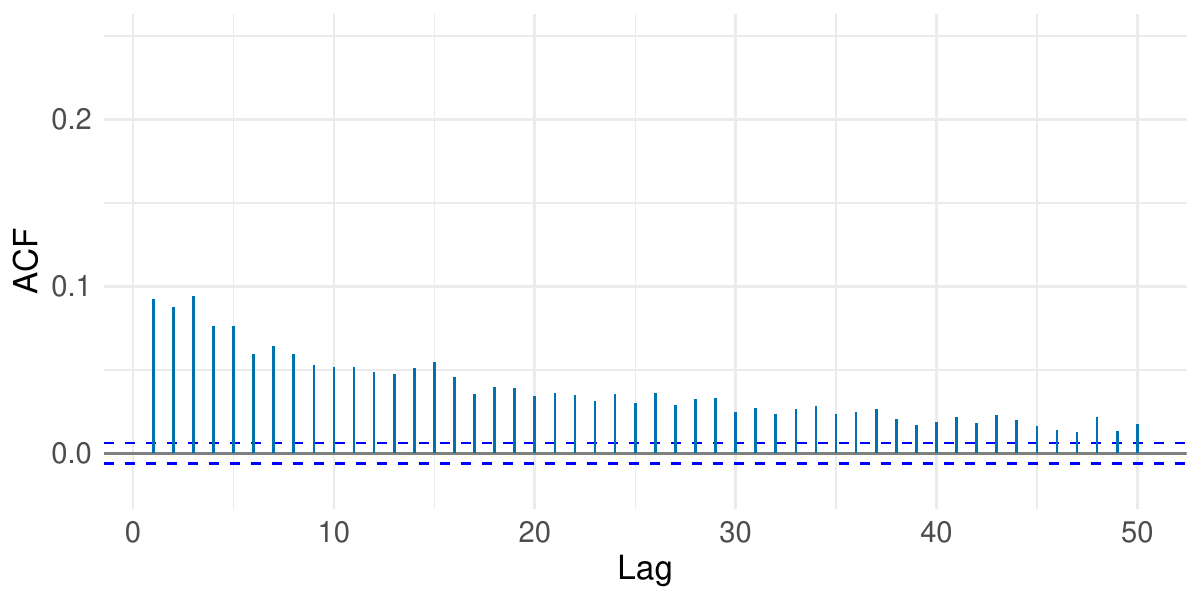}
		\caption{$\alpha = 0.07$ and $\beta = 0.1$}
	\end{subfigure}
	
	\caption{Autocorrelation functions of simulated inter-arrival times $\tau_n$ with Gamma residuals ($\mu_\varepsilon = 1$)}
	\label{Fig:acf_SE_gamma}
\end{figure}

To examine the dynamic and dependence properties implied by the model in Definition~\ref{Def:FRPP},
we simulate inter-arrival times using Gamma-distributed residuals.
The autocorrelation structure of the model is primarily governed by the relative magnitude
of the self-excitation parameter $\alpha$ and the decay rate $\beta$.
Provided the process operates in the stationary regime
(i.e., $\alpha < \beta \mu_\varepsilon$),
the degree of temporal persistence increases as the ratio $\alpha / \beta$ approaches unity.
This ratio controls how strongly and how persistently past events influence future
conditional intensities.

Figure~\ref{Fig:acf_SE_gamma} illustrates this mechanism by contrasting two parameter configurations.
When $\alpha = 0.095$ and $\beta = 0.1$, the conditional intensity $\Lambda_n$ decays slowly following an excitation,
which translates into pronounced duration clustering and a slowly decaying autocorrelation function of $\tau_n$.
In contrast, setting $\alpha = 0.07$ with the same decay rate leads to a faster dissipation of excitation effects,
resulting in weaker dependence and a more rapidly decreasing autocorrelation structure.

\section{Empirical analysis}\label{Sect:empirical}

\subsection{Data}

The point process models developed in Section~\ref{Sec:model} are applied to
ultrahigh-frequency limit order book (LOB) data for a liquid financial asset.
In such data, observed durations between events
(e.g., order submissions, trades, and cancellations)
exhibit substantial heterogeneity and variability.
The empirical distribution of inter-arrival times is typically heavy-tailed,
with extremely short durations interspersed with occasional but very long waiting times.
As a result, the data deviate markedly from the exponential distribution
commonly assumed in simpler benchmark models,
such as the exponential Hawkes process or the exponential autoregressive duration model.

In the empirical analysis, the duration (or inter-arrival time) is defined as
the elapsed time required for the mid-price to move by one tick.
The mid-price is computed as the average of the best bid and best ask prices,
and therefore changes by half a tick whenever either side of the book moves by one tick.
Consequently, a full one-tick movement of the mid-price typically results from
two consecutive half-tick changes in the same direction.
For example, a sequence of half-tick movements such as
up–down–up–up leads to a net upward movement of one tick,
thereby triggering a new event.

Table~\ref{tab:descriptive_stats} reports descriptive statistics for the AAPL
durations observed in December 2022, constructed using the one-tick
(two half-tick) mid-price movement criterion.
The data exhibit pronounced over-dispersion, with a standard deviation
($0.9937$) more than twice the mean ($0.4801$).
The resulting over-dispersion ratio of $2.0698$ indicates a substantial
departure from the Poisson benchmark.
Moreover, the large gap between the median ($0.1281$) and the mean,
together with a high skewness value ($5.7536$),
suggests a distribution that is strongly right-skewed.
Finally, the extremely large kurtosis ($61.4619$) provides clear evidence
of heavy-tailed behavior in the duration distribution.

\begin{table}[t]
	\centering
	\caption{Descriptive statistics of AAPL durations in seconds (December 2022)}
	\label{tab:descriptive_stats}
	\begin{tabular}{lr}
		\hline
		Statistic & Value \\
		\hline
		Number of Observations & 1,023,480 \\
		Mean & 0.4801 \\
		Standard Deviation & 0.9937 \\
		Minimum & $3.08\times10^{-7}$ \\
		Median & 0.1281 \\
		Maximum & 38.7261 \\
		Skewness & 5.7536 \\
		Kurtosis & 61.4619 \\
		Over-dispersion ($SD/Mean$) & 2.0698 \\
		\hline
	\end{tabular}
\end{table}

Figure~\ref{fig:dur_hist} compares the empirical duration distribution with a fitted exponential density.
The histogram shows a high concentration of very short durations near zero.
While the heavier right tail is not fully visible from the histogram due to scale compression,
it is confirmed by the descriptive statistics and systematic deviations from the exponential benchmark,
motivating the use of flexible residual distributions.

\begin{figure}[t]
	\centering
	\includegraphics[width=0.8\textwidth]{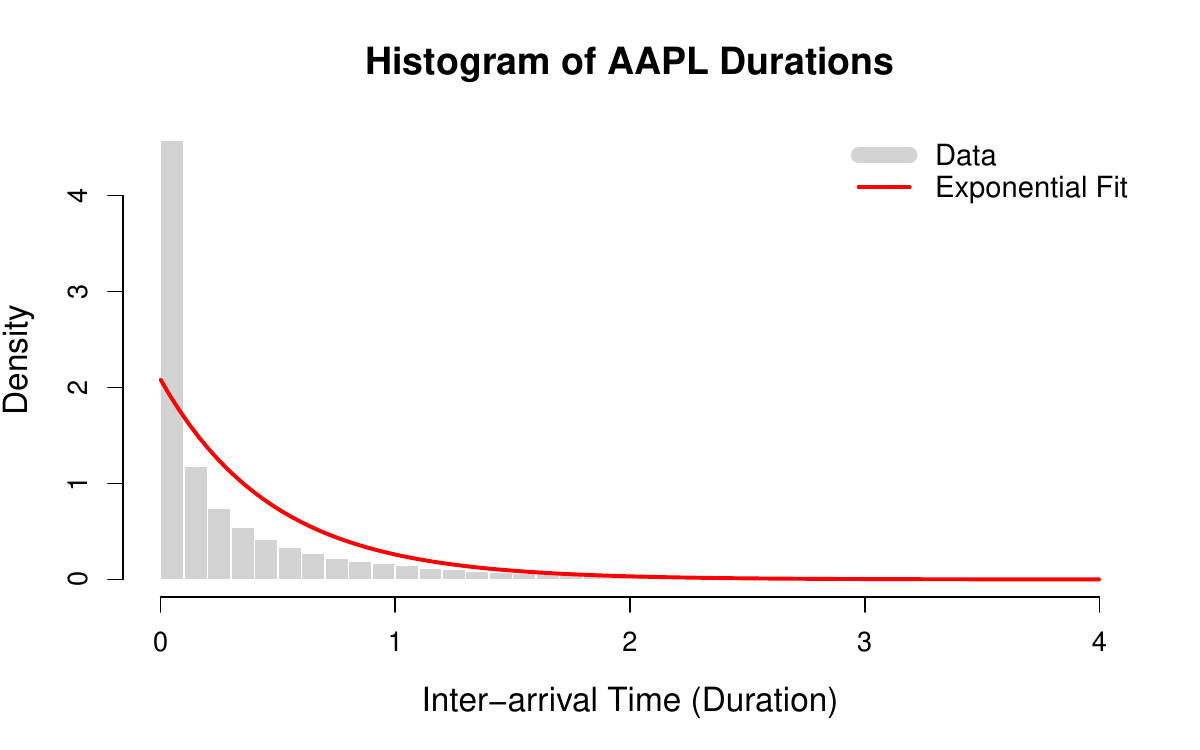}
	\caption{Histogram of AAPL durations with fitted exponential density}
	\label{fig:dur_hist}
\end{figure}

\subsection{Residual distributions}

To model the distributional properties of the duration innovations $\varepsilon$, we consider three candidate distributions:
the Gamma, generalized Gamma, and Burr Type XII distributions.
For identification and parsimony, all distributions are parameterized to satisfy $\mathbb{E}[\varepsilon]=1$.

\paragraph{Gamma distribution.}
As a benchmark, we employ the Gamma distribution.
Imposing the unit-mean restriction yields a single shape parameter $\kappa>0$, with scale fixed at $1/\kappa$.
The probability density function (PDF) is
\begin{equation}
	f(x; \kappa) = \frac{\kappa^{\kappa}}{\Gamma(\kappa)} x^{\kappa-1} \exp(-\kappa x), \quad x > 0.
\end{equation}
Under this parameterization, $\kappa$ controls dispersion, and the exponential distribution is recovered when $\kappa=1$.

\paragraph{Generalized Gamma distribution.}
To allow for greater flexibility in tail behavior, we also consider the generalized Gamma distribution.
For shape parameters $d,p>0$, the PDF is
\begin{equation}
	f(x; d, p) = \frac{p}{a^{d}\Gamma(d/p)} x^{d-1} \exp\left[-\left(\frac{x}{a}\right)^{p}\right], \quad x > 0,
\end{equation}
where the scale parameter $a$ is chosen to satisfy $\mathbb{E}[\varepsilon]=1$:
\begin{equation}
	a = \frac{\Gamma(d/p)}{\Gamma((d+1)/p)}.
\end{equation}
Here, $d$ governs the behavior near the origin, while $p$ controls tail thickness, allowing the distribution to accommodate varying degrees of leptokurtosis.

\paragraph{Burr Type XII Distribution.}
The Burr Type XII distribution is a highly flexible family capable of capturing a wide range of heavy-tailed behaviors \citep{burr1942cumulative}. 
It is characterized by two shape parameters $s_1, s_2 > 0$ and a scale parameter $c > 0$. 
Its probability density function (PDF) and hazard function $h(x)$ are given by
\begin{equation}
	f(x; s_1, s_2, c) = \frac{s_1 s_2}{c} \left( \frac{x}{c} \right)^{s_2 - 1} \left[ 1 + \left( \frac{x}{c} \right)^{s_2} \right]^{-(s_1 + 1)}, \quad x > 0
\end{equation}
\begin{equation}
	h(x; s_1, s_2, c) = \frac{s_1 s_2}{c} \frac{(x/c)^{s_2 - 1}}{1 + (x/c)^{s_2}}, \quad x > 0.
\end{equation}
The mean exists when $s_1 s_2 > 1$.
Imposing $\mathbb{E}[\varepsilon]=1$ yields
\begin{equation}
	c = \left[ s_1 \mathrm{B} \left( s_1 - \frac{1}{s_2}, 1 + \frac{1}{s_2} \right) \right]^{-1},
\end{equation}
where $\mathrm{B}(\cdot, \cdot)$ denotes the beta function. 
This distribution exhibits a power-law tail and is well suited for modeling highly skewed durations.

\subsection{Dynamic structures}

To conduct a comprehensive empirical comparison, we consider multiple combinations of latent dynamics and residual distributions.
Specifically, three classes of duration models are examined:
\begin{enumerate}
	\item \textbf{Self-Exciting (SE) models.}
	These models follow the dynamics in Definition~\ref{Def:FRPP}.
	The class includes SE-Exp \citep{Hawkes1}, SE-Gamma, SE-gGamma (generalized Gamma), and SE-Burr specifications.
	
	\item \textbf{ACD models.}
	These follow the linear update rule in Definition~\ref{Def:ACD}.
	We consider ACD-Exp \citep{Engle1998}, ACD-Gamma, ACD-gGamma, and ACD-Burr \citep{grammig2000non}.
	
	\item \textbf{Log-ACD models.}
	These employ the logarithmic update rule in Definition~\ref{Def:logACD}, which ensures positivity of the latent process without additional parameter constraints.
	This class includes logACD-Exp, logACD-Gamma, logACD-gGamma, and logACD-Burr models.
\end{enumerate}

Table~\ref{tab:estimation_results_all} reports the estimation results for above model specifications using AAPL duration data from December~16,~2022.
By combining three latent dynamics with four residual distributions, we assess their relative ability to capture ultra-high-frequency duration dynamics.
All parameter estimates are statistically significant; standard errors are reported in parentheses.

\begin{table}[!htbp]
	\centering
	\caption{Parameter estimates for AAPL durations (2022-12-16). Standard errors are reported in parentheses.}
	\label{tab:estimation_results_all}
	\renewcommand{\arraystretch}{1.3}
	\begin{tabular}{lcccc}
		\hline
		Model & $\mu$  & $\alpha$ & $\beta$   & Distribution parameters \\ 
		\hline
		SE-Exp 
		& 1.4768 & 90.9600 & 285.9000 &  \\
		& (0.0026) & (0.0051) & (0.0341) & \\

		SE-Gamma
		& 0.2712 & 0.0939 & 0.1068
		& $\kappa =0.3511$ (0.0018) \\
		& (0.0231) & (0.0038) & (0.0048) & \\
		
		SE-gGamma
		& 0.3682 & 0.1099 & 0.1319
		& $d=0.3340$ (0.0022) \\
		& (0.0538) & (0.0071) & (0.0116) & $p=1.2501$ (0.0286)\\
		
		SE-Burr 
		& 0.0327 & 0.0634 & 0.0743
		& $s_1=273.95$ (0.1444) \\
		& (0.0163) & (0.0024) & (0.0030) & $s_2=0.4765$ (0.0018) \\
		
		\hline
		& $b_0$ &  $a$ & $b_1$ & \\
		
		\hline
		ACD-Exp
		& $9.05\times10^{-5}$ & 0.0476 & 0.9523 &  \\
		& ($6.01\times10^{-6}$) & (0.0015) & (0.0012) & \\

		ACD-Gamma
		& $8.97\times10^{-5}$ & 0.0476 & 0.9524
		& $\kappa =0.3486$ (0.0018)\\
		& ($2.87\times10^{-5}$) & (0.0019) & (0.0018) &  \\
		
		ACD-gGamma
		& $4.68\times10^{-5}$ & 0.0481 & 0.9519
		& $d=0.3536$ (0.0025)\\
		& ($2.41\times10^{-5}$) & (0.0020) & (0.0019) &$p=0.9415$ (0.0151) \\
		
		ACD-Burr
		& $9.11\times10^{-5}$ & 0.0601 & 0.9399
		& $s_1=5.4367$ (0.0556) \\
		& ($1.04\times10^{-5}$) & (0.0029) & (0.0030) & $s_2=0.5066$ (0.0012) \\
		
		\hline
		logACD-Exp
		& 0.0270 & 0.0212 & 0.9664 &  \\
		& (0.0009) & (0.0007) & (0.0011) & \\
		
		logACD-Gamma
		& 0.0270 & 0.0212 & 0.9664
		& 	$\kappa=0.3440$ (0.0017)\\
		& (0.0015) & (0.0012) & (0.0019) &  \\
		
		logACD-gGamma
		& 0.0369 & 0.0266 & 0.9604
		& $d=0.3698$ (0.0030) \\
		& (0.0021) & (0.0015) & (0.0022) & $p=0.7964$ (0.0146) \\
		
		logACD-Burr
		& 0.0678 & 0.0392 & 0.9486
		& $s_1=363.29$(0.3707) \\
		& (0.0033) & (0.0019) & (0.0026) & $s_2=0.4726$ (0.0017) \\
		
		\hline
	\end{tabular}
\end{table}

Notably, the SE-Exp model produces disproportionately large estimates for $\mu$, $\alpha$, and $\beta$.
This behavior reflects the inability of the exponential residual to accommodate the pronounced heterogeneity and over-dispersion observed in the data.
As a result, the model compensates by inflating the baseline intensity and excitation parameters.
In contrast, SE models incorporating flexible residuals (Gamma, gGamma, and Burr) exhibit significantly smaller and more stable parameter estimates, better reflecting the underlying dynamics.

For ACD-type models, Table~\ref{tab:estimation_results_all} indicates strong persistence,
with $a + b_1$ close to unity across all specifications.
The estimated shape parameter $\kappa \approx 0.34$ confirms significant over-dispersion relative to the exponential benchmark. 
Similarly, the non-unit tail parameters in the generalized Gamma and Burr specifications
underscore the importance of flexible residual distributions
for capturing both the high density of short durations and the heavy right tail observed in the data.

\subsection{Intraday estimation}

To account for the diurnal patterns and time-varying statistical characteristics of durations, we adopt a rolling-window estimation scheme.
Specifically, the model parameters are updated locally for each estimation window using the 5,000 most recent mid-price events (averaging approximately 45 minutes). 
This dynamic intraday estimation allows the model to adapt to changing market conditions and effectively captures time-of-day effects without requiring a deterministic seasonality component. 
Based on these localized estimates, we compute the expected interarrival times for the subsequent 100 out-of-sample events using Eqs.~\eqref{Eq:cond_tau_ACD} and~\eqref{Eq:cond_tau_SE}.

Forecasting requires reconstruction of the latent state associated with each model.
For ACD-type models, the conditional duration process $X_n$ is inferred recursively using the estimated parameters
$\hat b_0$, $\hat a$, and $\hat b_1$.
For the self-exciting flexible residual point process, the latent state process $\Lambda_n$ is reconstructed using $\hat\mu$, $\hat\alpha$, and $\hat\beta$.

For models with flexible residual distributions (Gamma, generalized Gamma, or Burr),
conditional expectations are computed using the corresponding estimated shape parameters
($\hat{\kappa}$, $(\hat d,\hat p)$, or $(\hat s_1,\hat s_2)$) together with the inferred latent states.
During the out-of-sample period, latent processes are updated recursively as new events arrive,
ensuring that each forecast is based on the most recent state information.

This procedure is applied sequentially throughout the trading day, advancing the window by 100 events at each step.
As a result, forecasts are produced for almost the entire sample except for the initial 5{,}000 events used for the first estimation.
Predicted durations are then compared with realized interarrival times to assess predictive performance.

As a benchmark, we consider a log-ACI(1,1) model in Definition~\ref{Def:logACI} for durations.
In addition, 
to accommodate long-memory dynamics in financial durations,
we consider the semiparametric fractionally integrated (FI) log-ACD(1,$d$,1) model proposed by
\citet{FengZhou2015}.
The model assumes
\begin{equation}
	\log \tau_n = \mu + z_n = \mu + \log \psi_n + \varepsilon_n,
\end{equation}
where $z_n := \log \psi_n + \varepsilon_n$.
The latent process $z_n$ is assumed to satisfy
\begin{equation}
	(1 - \phi B)(1-B)^d z_n = (1 + \theta B)\varepsilon_n
\end{equation}
with parameters satisfying
\[
d \in (0,1/2), \quad |\phi|<1, \quad |\theta|<1
\]
where $d$ governs the hyperbolic decay of the autocorrelation function.
Note that the FI log-ACD model falls outside the general framework in Definition~\ref{Def:general}.

The predictive evaluation is conducted using mid-price data from December~1 to 31, 2022.
For each trading day, model parameters are estimated using the rolling-window procedure described above, and the estimation–prediction cycle is repeated sequentially across the entire month.

\subsection{Performance result}

Table~\ref{Table:RMSE} summarizes the out-of-sample predictive performance of the various models averaged over the sample period of December 2022. 
The performance was evaluated using the relative root mean squared 
error (rRMSE) between the predicted and actual interarrival times, 
the coefficient of determination (R-squared), the log-ratio root 
mean squared error (LRMSE), defined as 
$\sqrt{n^{-1}\sum_t(\log\hat{\tau}_t-\log\tau_t)^2}$,
which penalizes 
over- and under-prediction symmetrically on the ratio scale, 
the Kolmogorov--Smirnov statistic for residuals, and the mean 
squared Wasserstein distance (W) between the empirical and 
model-based distributions of residuals.

\begin{table}[t]
	\centering
	\caption{Out-of-sample prediction and goodness-of-fit results for various models using AAPL data in December 2022.}
	\label{Table:RMSE}
	\begin{tabular}{lccccc}
		\hline
		Model & rRMSE & R-squared & LRMSE & KS & W \\
		\hline
		
		SE-Exp (Hawkes) & 1.8386 & 0.0782 & 2.8907 & $0.1381^{***}$ & 0.0052 \\
		SE-Gamma        & 1.8542 & 0.1374 & 2.8863 & $0.0507^{***}$ & 0.0009 \\
		SE-gGamma       & 1.8739 & 0.0521 & 2.8720 & $0.0486^{***}$ & 0.0008 \\
		SE-Burr         & 1.8033 & 0.1134 & 2.9127 & $0.0659^{***}$ & 0.0016 \\
		
		\hline
		ACD-Exp         & 1.8544 & 0.1371 & 2.8832 & $0.2020^{***}$ & 0.0147 \\
		ACD-Gamma       & 1.8533 & 0.1382 & 2.8815 & $0.0236^{***}$ & 0.0002 \\
		ACD-gGamma      & 1.8555 & 0.1372 & 2.8851 & $0.2077^{***}$ & 0.0159 \\
		ACD-Burr        & 1.8589 & 0.1329 & 2.9190 & $0.0589^{***}$ & 0.0012 \\
		
		\hline
		logACD-Exp      & 1.8639 & 0.1284 & 2.8629 & $0.2053^{***}$ & 0.0169 \\
		logACD-Gamma    & 1.8640 & 0.1282 & 2.8782 & $0.0209^{***}$ & 0.0001 \\
		logACD-gGamma   & 1.8779 & 0.1151 & 2.9080 & $0.0186^{***}$ & 0.0002 \\
		logACD-Burr     & 1.9226 & 0.0724 & 2.9506 & $0.0386^{***}$ & 0.0005 \\
		
		\hline
		logACI          & 1.8428 & 0.1331 & 2.8690 & $0.2007^{***}$ & 0.0145 \\
		\hline
		
		FI-logACD       & 1.8815 & 0.1117 & 2.8768 & $0.1738^{***}$ & 0.0147 \\
		\hline
	\end{tabular}\\[4pt]
	{\footnotesize $^{***}$ significant at the 1\% level.}
\end{table}

For SE models, the SE-Burr specification improves upon the exponential SE model 
in terms of rRMSE and also yields a higher R-squared, 
indicating a noticeable gain in predictive accuracy.
In contrast, the SE-Gamma and SE-gGamma models do not provide improvements in rRMSE, although SE-Gamma achieves the highest R-squared within the SE class.

From the goodness-of-fit viewpoint, however, all flexible specifications substantially reduce the KS and W statistics relative to SE-Exp, with SE-Gamma and SE-gGamma showing the strongest improvements.
This suggests that while the Gamma-type residuals do not translate into better point prediction accuracy, they provide a markedly better description of the distributional properties of the interarrival times.

Overall, these results indicate that the choice of residual distribution affects different aspects of model performance in distinct ways.
The Burr specification delivers the balanced improvement by enhancing both predictive accuracy and distributional fit, whereas the Gamma-based specifications mainly improve the latter.

For the ACD models, in terms of predictive accuracy, all ACD variants deliver relatively similar rRMSE and R-squared values, indicating that replacing the exponential distribution with Gamma-, generalized Gamma-, or Burr-based alternatives does not lead to a meaningful improvement in forecasting performance.

In contrast, the goodness-of-fit results show substantial differences across distributions.
The ACD-Gamma model achieves by far the lowest KS and W statistics among the ACD specifications, indicating a much better match to the empirical distribution of interarrival times.
The ACD-Burr model also improves the fit relative to the exponential case, whereas the ACD-gGamma specification provides little or no gain in this respect.

Log-ACD models display a pattern very similar to that of the standard ACD specifications.
Their predictive accuracy is comparable to, but not better than, that of the ACD models, while the goodness-of-fit improves markedly when flexible residual distributions—especially the Gamma and generalized Gamma—are used.

The log-ACI model delivers a competitive rRMSE among all specifications, while its R-squared remains at a moderate level.
However, its relatively large KS and W statistics reveal a poor distributional fit compared with the other models with flexible residuals.

Finally, although the FI-logACD model introduces long-memory dynamics, such extensions do not systematically improve out-of-sample predictive performance in the present empirical setting.

Regarding LRMSE, the log-ACD family and FI-logACD tend to perform more favorably than the SE and ACD models, suggesting that log-scale specifications provide relatively better predictions for short durations. 
In contrast, SE-Burr achieves the lowest rRMSE, reflecting superior prediction of large duration events through the flexible Burr tail; the two metrics thus capture complementary aspects of predictive performance.
Overall, the results suggest that distributional flexibility 
is particularly beneficial within the self-exciting framework. 
In this setting, flexible residuals substantially improve the 
goodness of fit and, in the case of the Burr specification, 
are also associated with a measurable gain in predictive accuracy.

\begin{table}[t]
	\centering
	\caption{Diebold--Mariano test results for equal predictability
		against SE-Burr using AAPL data in December 2022.}
	\label{Table:DM}
	\begin{tabular}{lrrl@{\hspace{1.5em}}rrl}
		\hline
		& \multicolumn{3}{c}{\textit{Panel A: rRMSE loss}}
		& \multicolumn{3}{c}{\textit{Panel B: LRMSE loss}} \\
		\cmidrule(lr){2-4}\cmidrule(lr){5-7}
		Model & $\bar{d}$ & DM & & $\bar{d}$ & DM & \\
		\hline
		SE-Exp (Hawkes) & $-0.1237$ & $-51.860$ & $***$ & $0.0820$ & $19.930$  & $***$ \\
		SE-Gamma        & $-0.6955$ & $-26.185$ & $***$ & $0.3593$ & $231.708$ & $***$ \\
		SE-gGamma       & $-0.1094$ & $-15.152$ & $***$ & $0.3248$ & $158.967$ & $***$ \\
		ACD-Exp         & $-0.6965$ & $-26.507$ & $***$ & $0.3696$ & $146.059$ & $***$ \\
		ACD-Gamma       & $-0.6930$ & $-26.381$ & $***$ & $0.3798$ & $146.338$ & $***$ \\
		ACD-gGamma      & $-0.7078$ & $-26.863$ & $***$ & $0.3664$ & $132.577$ & $***$ \\
		ACD-Burr        & $-0.7160$ & $-27.753$ & $***$ & $0.2058$ & $75.114$  & $***$ \\
		logACD-Exp      & $-0.7303$ & $-27.513$ & $***$ & $0.4904$ & $163.340$ & $***$ \\
		logACD-Gamma    & $-0.7325$ & $-27.612$ & $***$ & $0.4353$ & $143.047$ & $***$ \\
		logACD-gGamma   & $-0.7492$ & $-28.214$ & $***$ & $0.3830$ & $119.394$ & $***$ \\
		logACD-Burr     & $-0.7837$ & $-30.229$ & $***$ & $0.2882$ & $82.581$  & $***$ \\
		logACI          & $-0.5715$ & $-21.642$ & $***$ & $0.3515$ & $136.822$ & $***$ \\
		FI-logACD       & $-0.8347$ & $-31.937$ & $***$ & $0.4372$ & $102.547$ & $***$ \\
		\hline
	\end{tabular}
\end{table}

To formally assess whether the predictive superiority of SE-Burr is statistically significant, we apply the modified Diebold--Mariano test of \citet{Diebold1995} as corrected by \citet{Harvey1997} for equal predictability between models.
Table~\ref{Table:DM} reports pairwise test results using SE-Burr as the benchmark under two loss functions.
Panel~A uses the squared relative error loss, $L_t = \{(\tau_t - \hat{\tau}_t)/\bar{\tau}\}^2$, consistent with the rRMSE criterion in Table~\ref{Table:RMSE}.
Panel~B uses the log-ratio squared error loss,
$L_t = (\log\hat{\tau}_t-\log\tau_t)^2$, 
consistent with LRMSE.
In both panels, $\bar{d} = n^{-1}\sum_t(L_t^{\text{SE-Burr}} - L_t^{\text{comp}})$ denotes the mean loss differential, so that $\bar{d} < 0$ indicates SE-Burr is the better-performing model, and $^{***}$ denotes significance at the 1\% level.

Under the rRMSE loss (Panel~A), $\bar{d} < 0$ for all competitor models, and the null hypothesis of equal predictability is rejected at the 1\% level in every case.
This confirms that SE-Burr achieves statistically significantly lower prediction error than all other specifications when evaluated on an absolute scale.

Under the LRMSE loss (Panel~B), the results are reversed: 
$\bar{d} > 0$ for all competitors, indicating that SE-Burr incurs significantly higher loss than every other model on the ratio scale.
As discussed above, this reflects the sensitivity of LRMSE to short durations, for which log-scale specifications have a natural advantage.
Taken together, the two panels confirm that the choice of loss function captures fundamentally different aspects of predictive performance: SE-Burr excels at predicting large duration events through its flexible Burr tail, while log-scale models provide
more accurate relative predictions across the full range of duration magnitudes.

\subsection{Goodness-of-fit}

To evaluate the distributional fit of each model, we construct probability–probability (P–P) plots comparing the model-implied conditional distribution of interarrival times with empirical observations. 
We evaluate the conditional cumulative distribution function $F_{\tau}$, as defined in Eq.~\eqref{Eq:F_tau} based on the estimated parameters.
Then, by inserting the observed interarrival times into the conditional CDF and plotting the resulting values against the uniform $[0,1]$ quantiles from the theoretical distribution, the P–P plot is generated. 
A model with good fit should produce points that lie close to the $45^\circ$ line.

\begin{figure}[t]
	\centering
	
	\begin{subfigure}[t]{0.4\textwidth}
		\includegraphics[width=\textwidth]{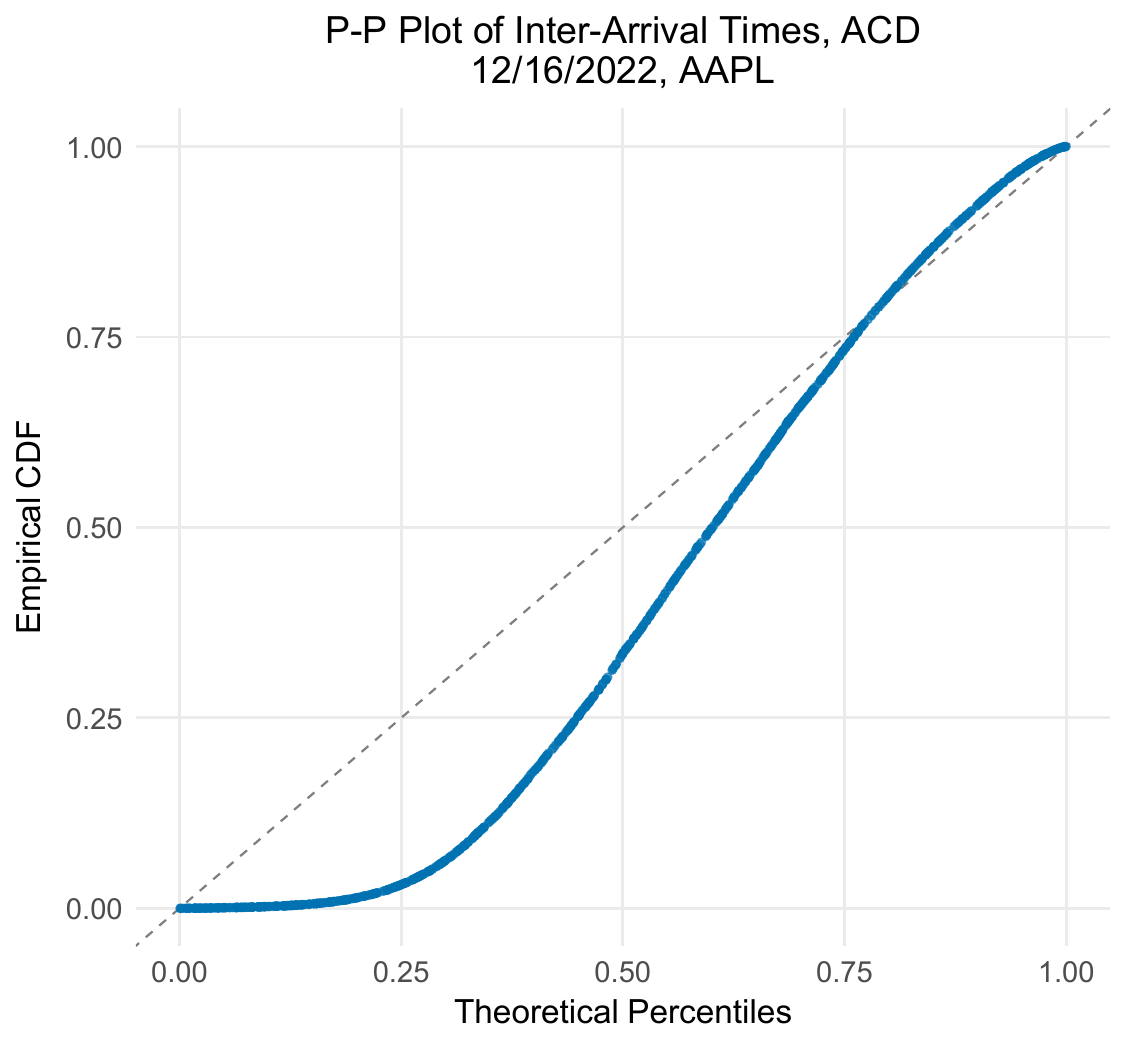}
		\caption{ACD-Exp}
	\end{subfigure}
	\quad\quad
	\begin{subfigure}[t]{0.4\textwidth}
		\includegraphics[width=\textwidth]{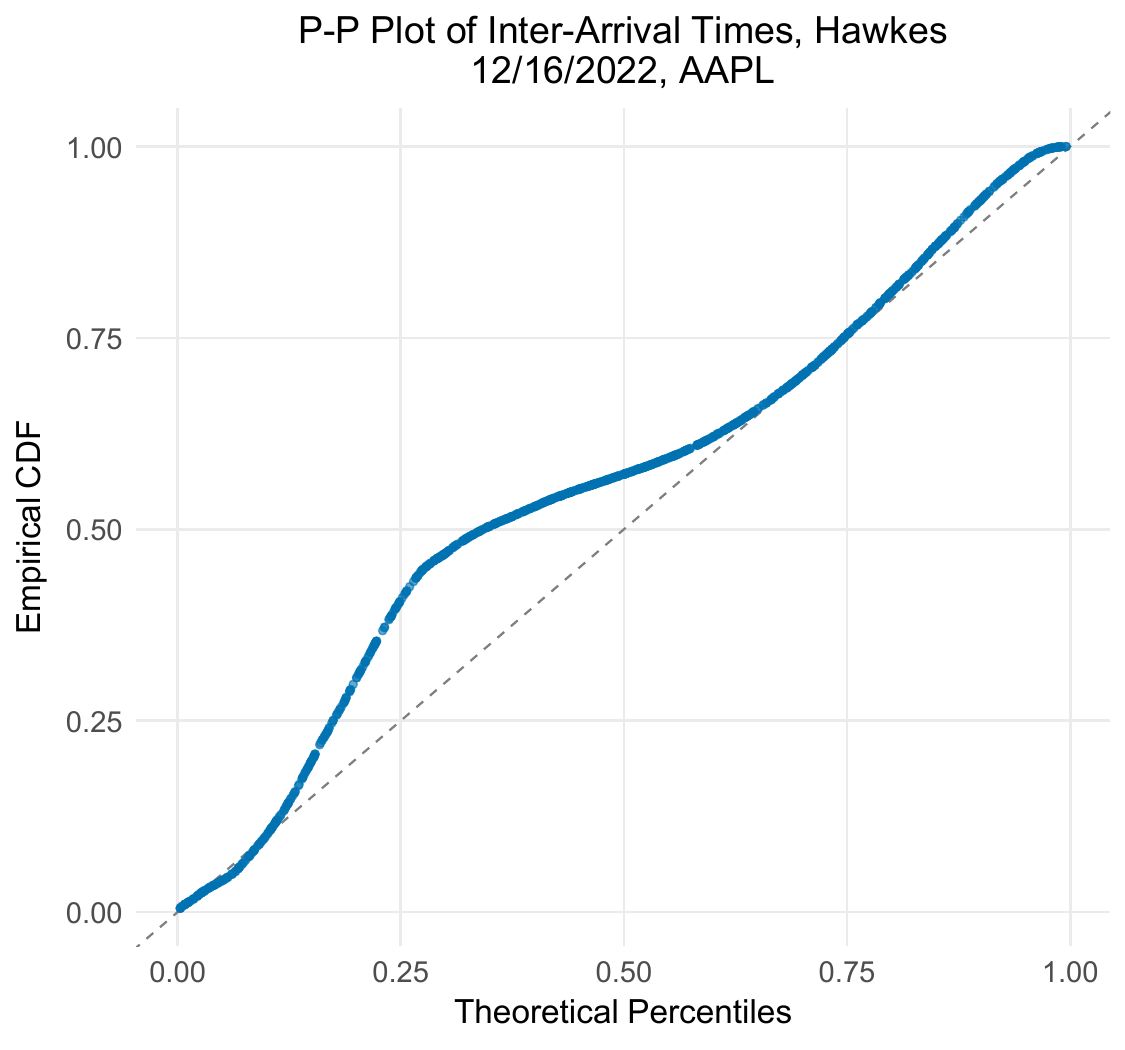}
		\caption{SE-Exp (Hawkes)}
	\end{subfigure}
	
	\vspace{0.5cm}
	
	\begin{subfigure}[t]{0.4\textwidth}
		\includegraphics[width=\textwidth]{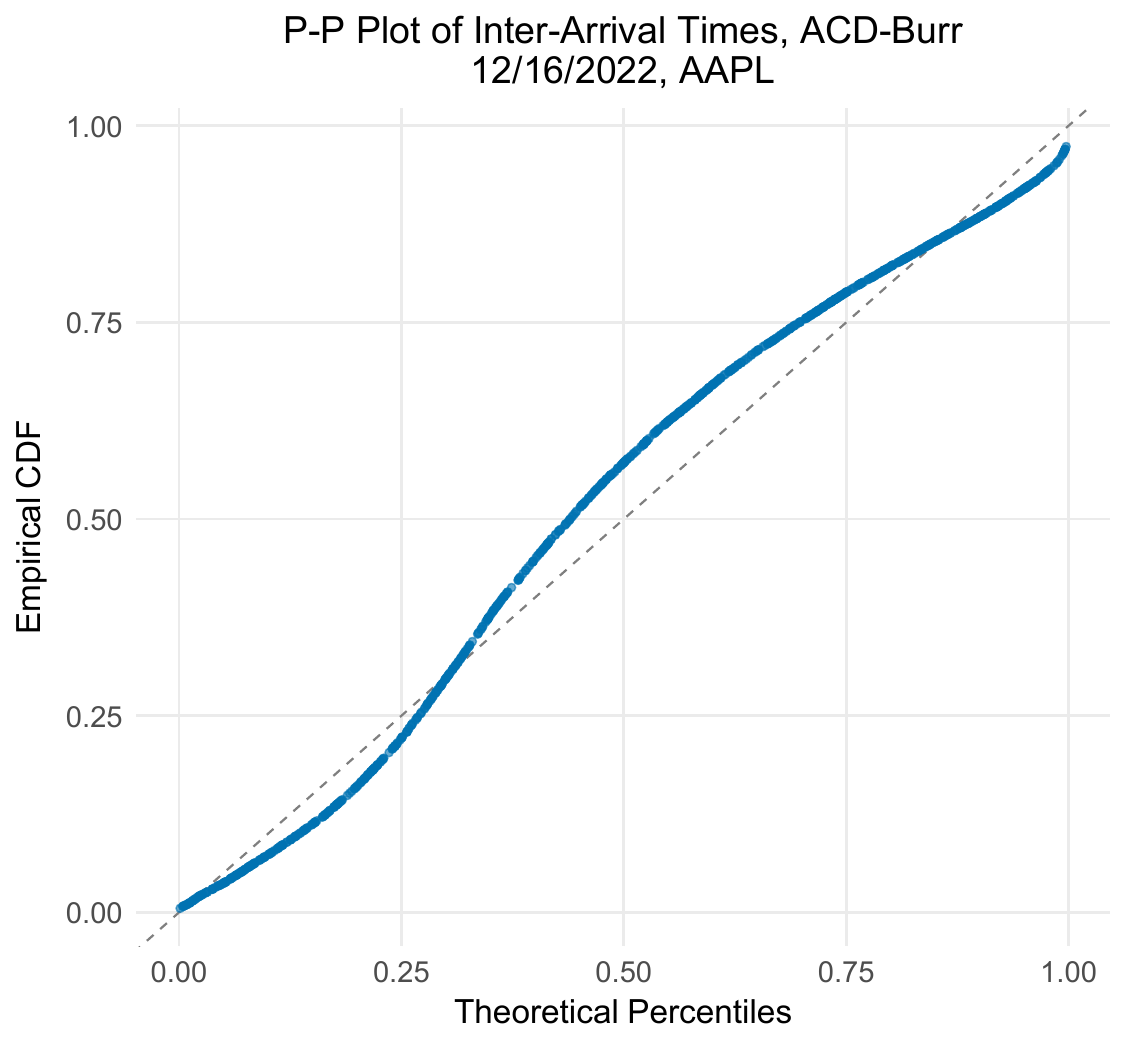}
		\caption{ACD–Burr}
	\end{subfigure}
	\quad\quad
	\begin{subfigure}[t]{0.4\textwidth}
		\includegraphics[width=\textwidth]{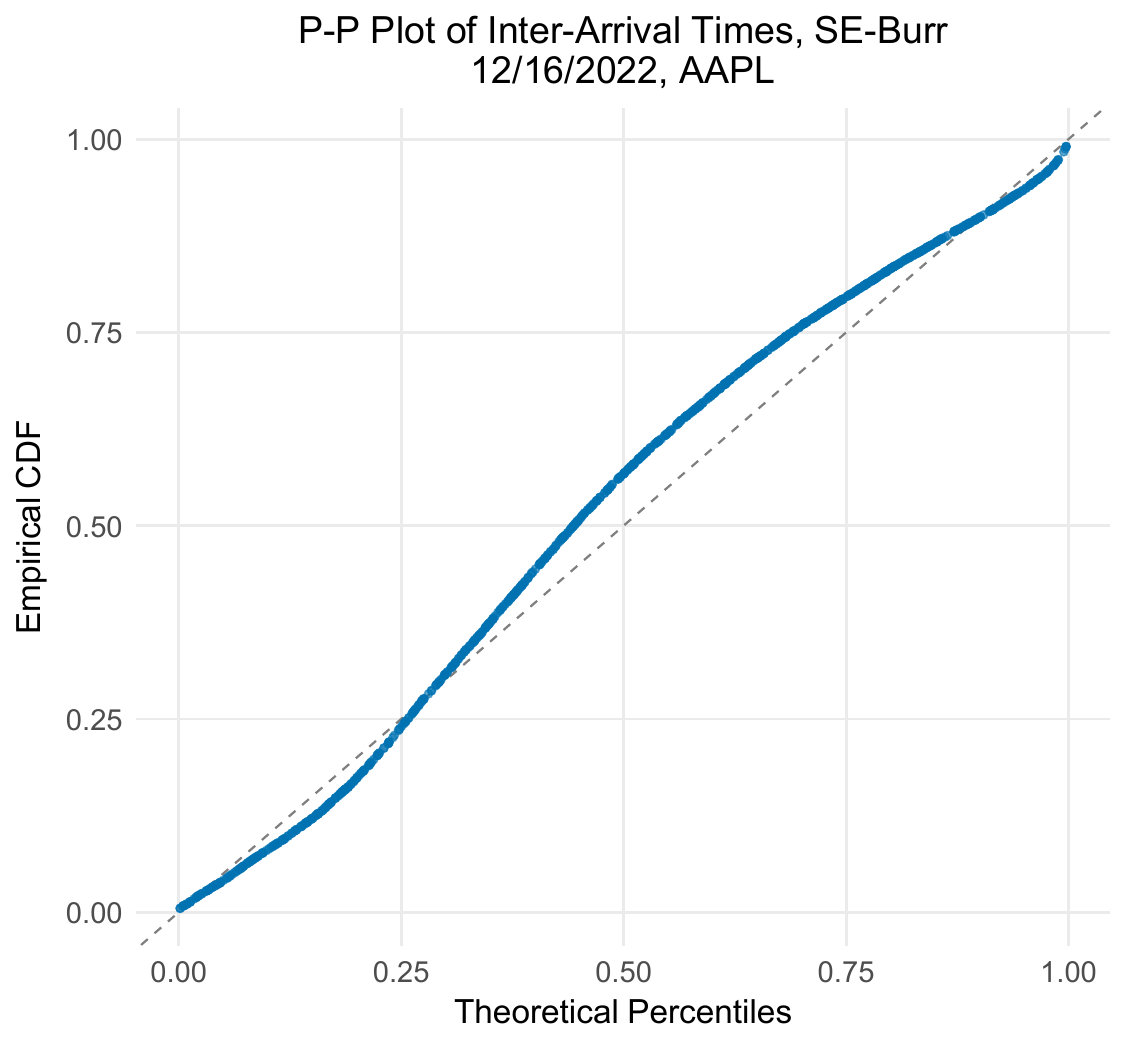}
		\caption{SE–Burr}
	\end{subfigure}
	
	\caption{P–P plots of interarrival times for four models using AAPL mid-price data on December 16, 2022.}
	\label{Fig:pp_plots}
\end{figure}

Figure~\ref{Fig:pp_plots} shows the P–P plots for each model, generated using AAPL mid-price data from December~16,~2022, based on rolling-window estimation with 5,000 observations.
The ACD model with an exponential interarrival distribution noticeably deviates from the diagonal, indicating limited flexibility in capturing the empirical distribution. 
The SE-Exp model performs better, but significant misalignment remains.  
In contrast, the ACD-Burr and SE-Burr models present substantially improved fit, with the SE-Burr model yielding the closest alignment to the diagonal line. 
This result suggests that combining a flexible heavy-tailed distribution, such as the Burr, with a self-exciting and decaying structure of $\Phi$ yields a more accurate representation of the interarrival time dynamics in high-frequency trading data.
The Gamma-family models produced similar results, but are omitted here for space considerations.

However, even the flexible distribution based models do not achieve a perfect fit, indicating that room for improvement remains, either in the specification of the residual distribution or in the functional form of the latent dynamics (e.g., the $\Psi$ or $\Phi$ structures) to better capture the complexities of the data. 
This is consistent with the classic finding of \citet{bauwens2004comparison} 
that capturing the conditional density of trade durations is 
empirically challenging even with the highly flexible Burr family, 
likely due to the clustering of very short trade durations in 
high-frequency data.

To further evaluate the goodness-of-fit of each model, we examine the distribution of exponential residuals, defined as the model-implied probability integral transforms of the observed interarrival times, as given in Eq.~\eqref{Eq:eres}. 
Under a correctly specified model, these exponential residuals should follow the standard exponential distribution with unit rate. 
Figure~\ref{Fig:exp_residuals} plots the histograms of the residuals for the ACD-Exp, SE-Exp, ACD-Burr and SE-Burr models.

\begin{figure}[t]
	\centering
	
	\begin{subfigure}[t]{0.45\textwidth}
		\includegraphics[width=\textwidth]{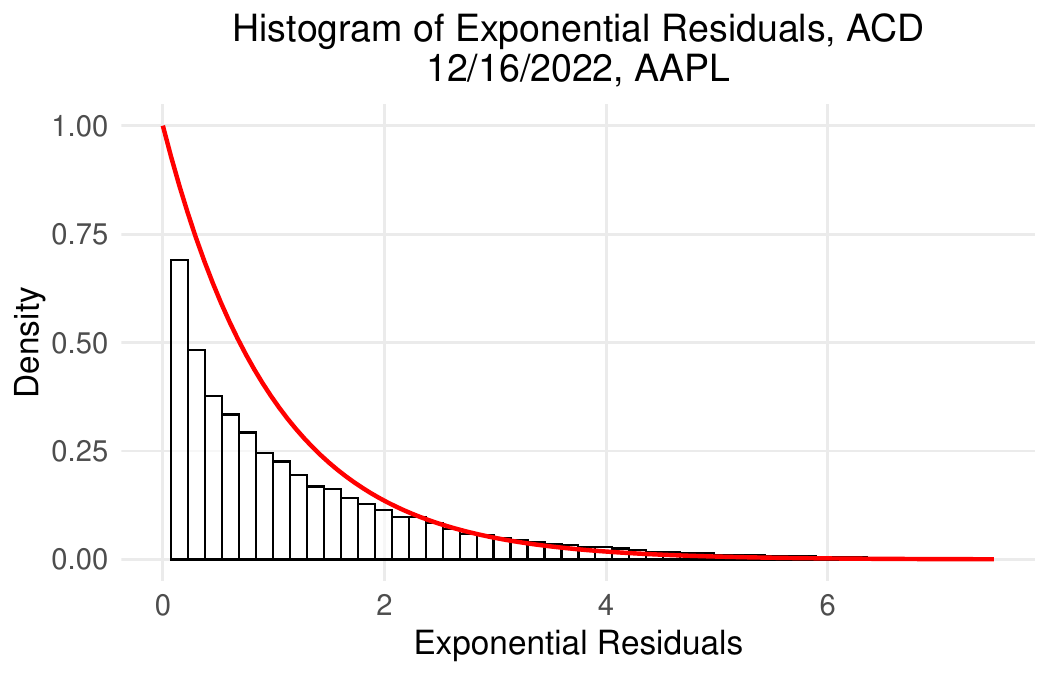}
		\caption{ACD-Exp}
	\end{subfigure}
	\quad
	\begin{subfigure}[t]{0.45\textwidth}
		\includegraphics[width=\textwidth]{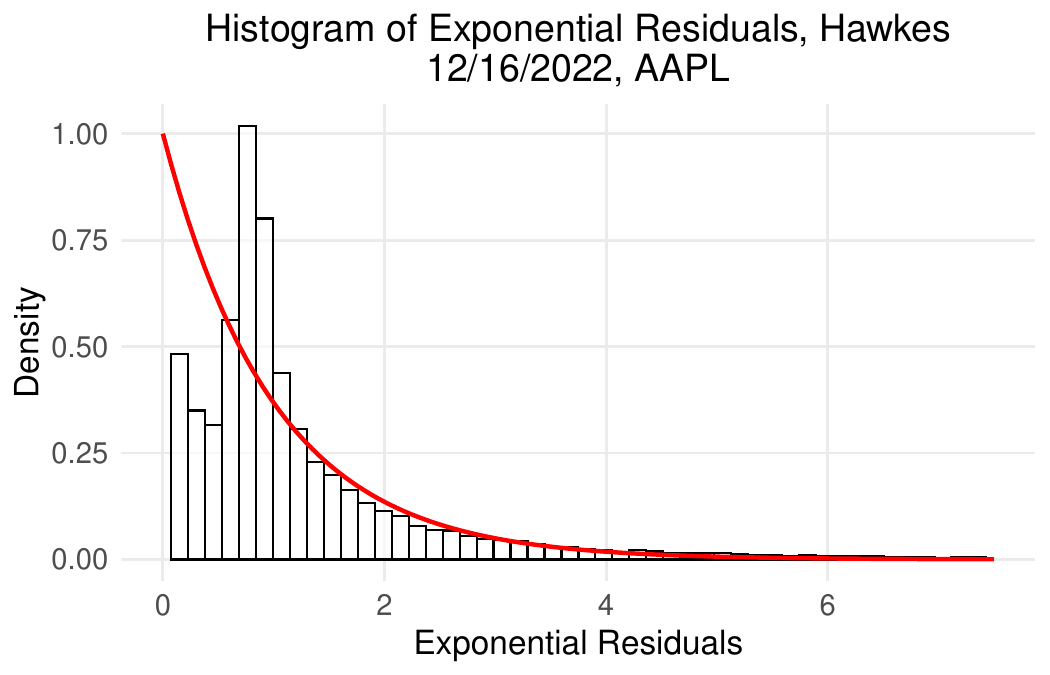}
		\caption{SE-Exp}
	\end{subfigure}
	
	\vspace{0.5cm}
	
	\begin{subfigure}[t]{0.45\textwidth}
		\includegraphics[width=\textwidth]{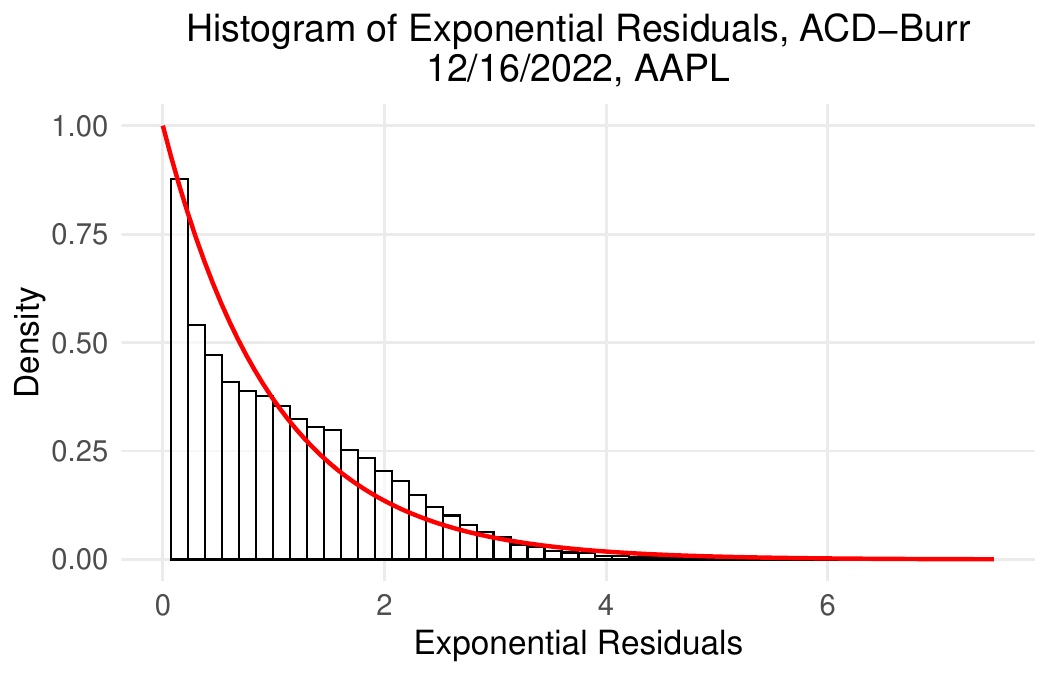}
		\caption{ACD–Burr}
	\end{subfigure}
	\quad
	\begin{subfigure}[t]{0.45\textwidth}
		\includegraphics[width=\textwidth]{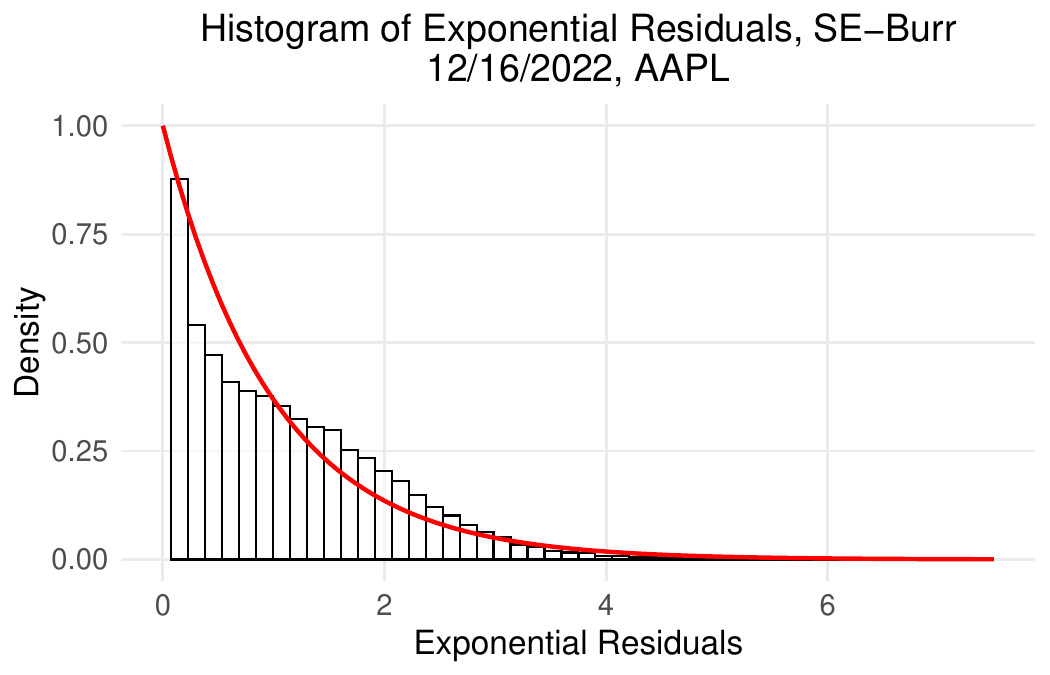}
		\caption{SE–Burr}
	\end{subfigure}
	
	\caption{Histograms of exponential residuals for four models using AAPL data on December 16, 2022.}
	\label{Fig:exp_residuals}
\end{figure}

The histograms of the exponential residuals from the ACD-Exp and SE-Exp (Hawkes) models present significant deviations from the unit exponential distribution. 
In contrast, the histograms of the residuals from the ACD-Burr and SE-Burr models are much closer in shape to the exponential density. 
These findings support the interpretation that the Burr-based model better fits the observed data.
The residuals from the Gamma-based model showed a similar pattern to those from the Burr-based model, but are omitted here for space considerations.
These findings are consistent with the P–P plot results, reinforcing the advantages of using heavy-tailed distributions in modeling high-frequency financial durations.

Nonetheless, the remaining deviations from the unit exponential 
distribution, particularly in the right tail, suggest that 
no model fully resolves the challenges documented in 
\citet{bauwens2004comparison}, and further refinements may be warranted.

Figure~\ref{Fig:Etau_dynamics} displays the inferred dynamics of the conditional expected interarrival times over the course of a trading day, with estimates obtained from each model using a rolling-window approach.
Despite the differences in model structure, all models display a shared intraday pattern. 
The expected interarrival times are shortest near the market open and close time, indicating periods of heightened price activity and trading intensity.
Conversely, the longest interarrival times occur between 2:00 and 3:00 p.m., suggesting a lull in market activity during that time window.

\begin{figure}[t]
	\centering
	
	\begin{subfigure}[t]{0.45\textwidth}
		\includegraphics[width=\textwidth]{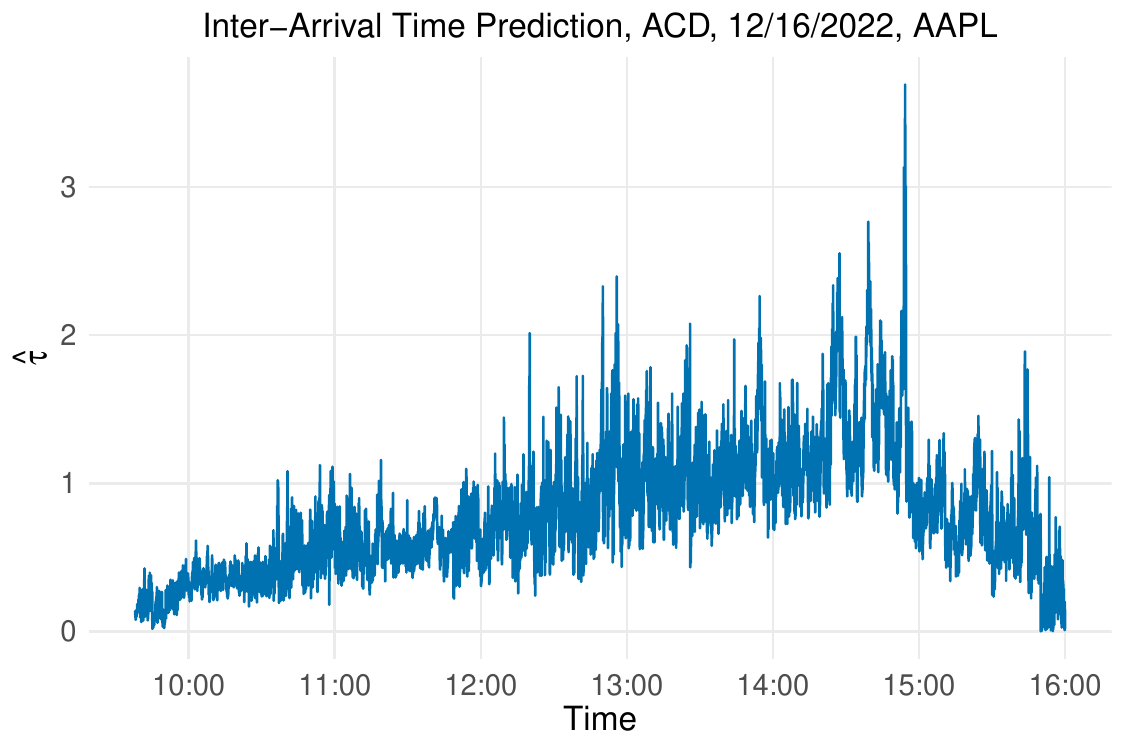}
		\caption{ACD-Burr}
	\end{subfigure}
	\quad
	\begin{subfigure}[t]{0.45\textwidth}
		\includegraphics[width=\textwidth]{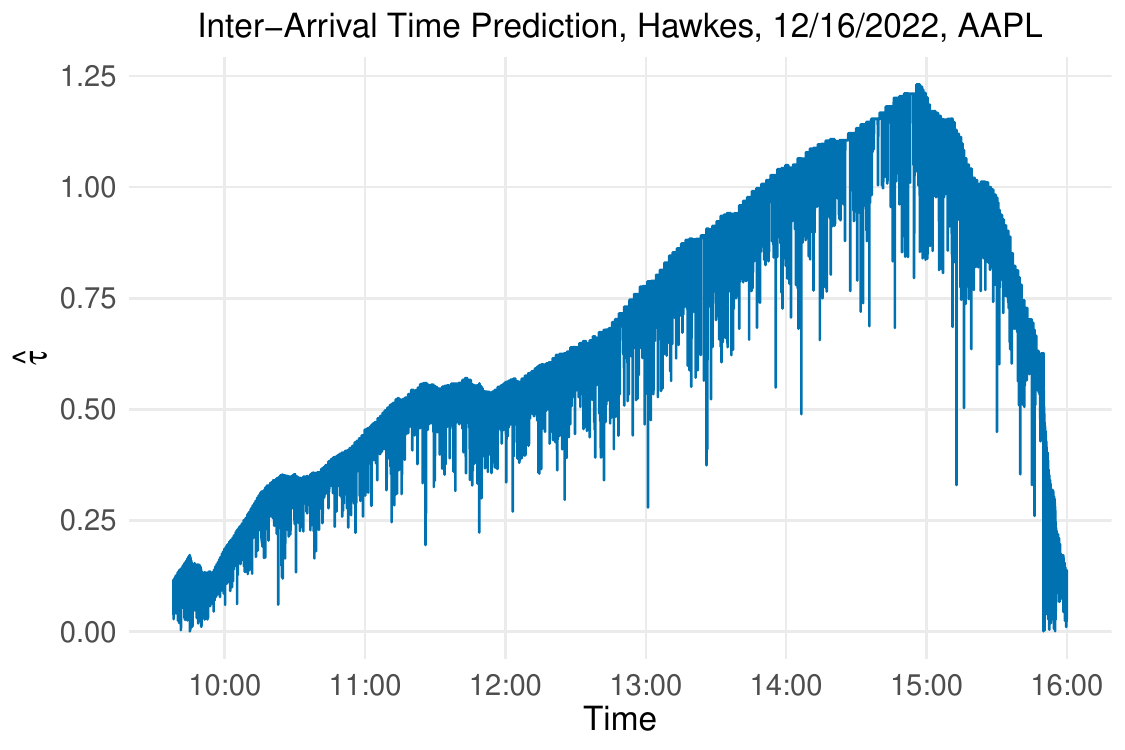}
		\caption{SE-Exp}
	\end{subfigure}
	
	\vspace{0.5cm}
	
	\begin{subfigure}[t]{0.45\textwidth}
		\includegraphics[width=\textwidth]{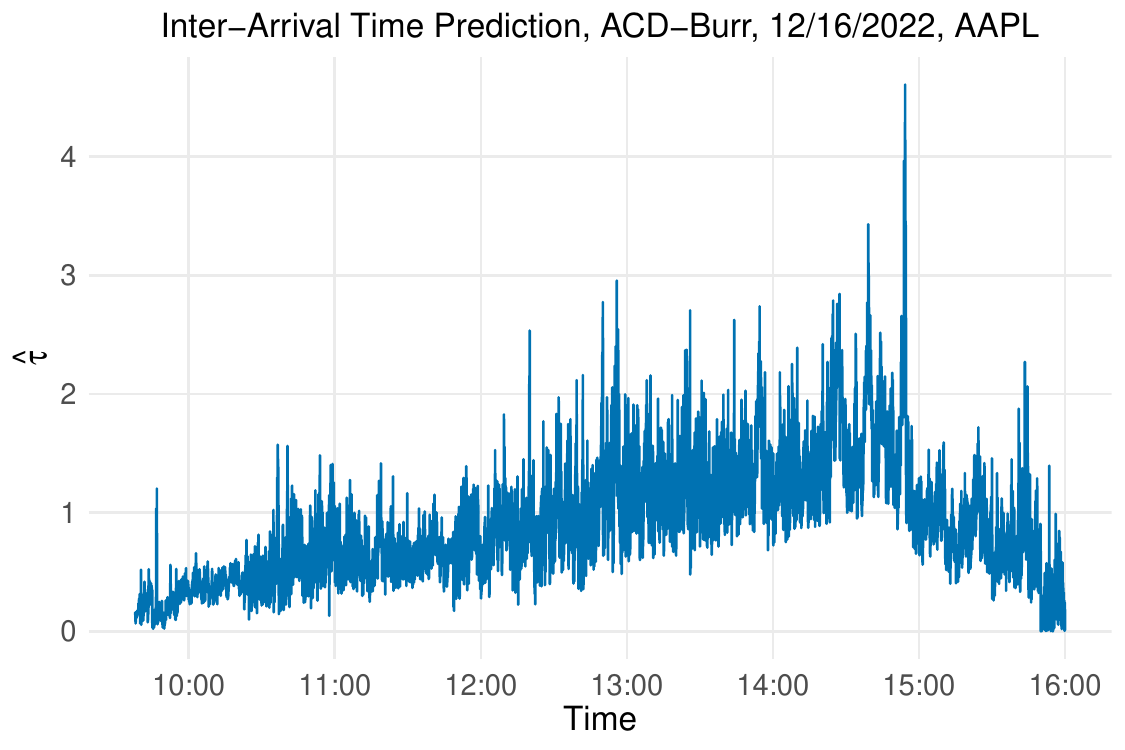}
		\caption{ACD–Burr}
	\end{subfigure}
	\quad
	\begin{subfigure}[t]{0.45\textwidth}
		\includegraphics[width=\textwidth]{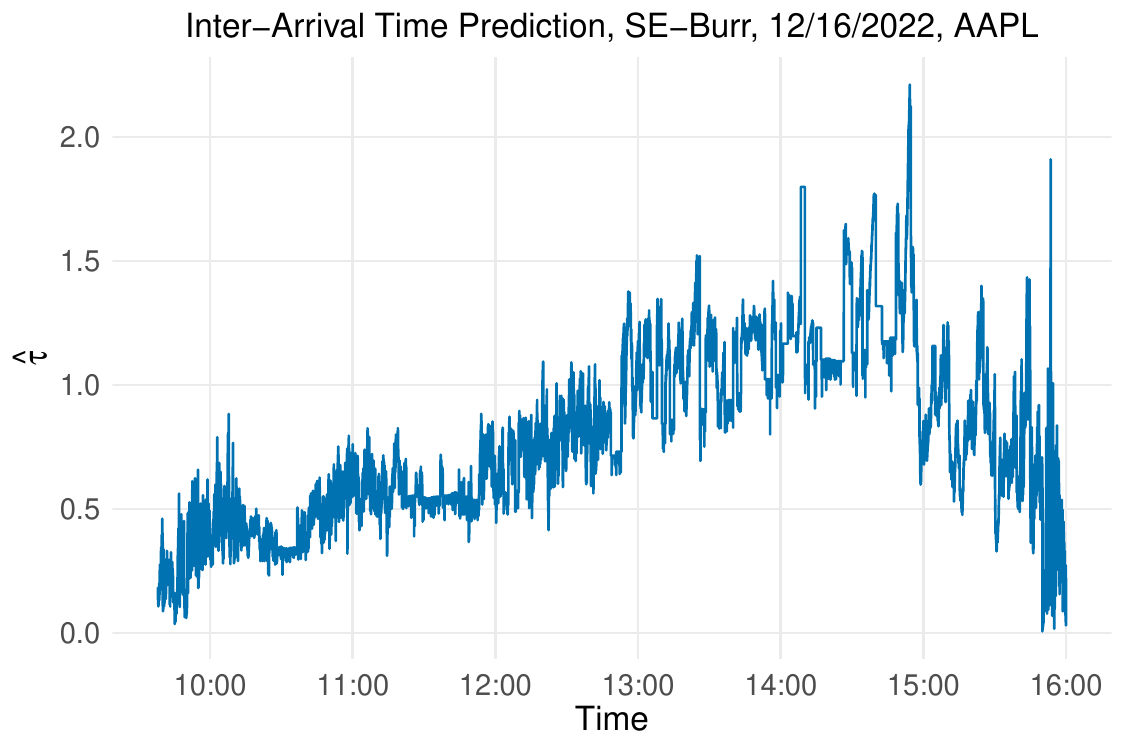}
		\caption{SE–Burr}
	\end{subfigure}
	
	\caption{Dynamics of the expected interarrival times predicted by four models, based on AAPL mid-price data on December 16, 2022.}
	\label{Fig:Etau_dynamics}
\end{figure}

Among the models, the SE-Exp (Hawkes) process reveals the most distinct dynamic pattern. 
The SE-Burr model presents the most stable dynamics, with relatively smooth, less volatile changes in the expected interarrival times throughout the day. 
These differences reflect how each model captures intraday market dynamics and trading-intensity variability.

The comparison in Table~\ref{Table:RMSE} of the ACD-Burr model and SE-Burr residual point process revealed that the latter provides slightly better predictive performance based on the rRMSE and  R-squared values. 
In contrast, the ACD-Burr model exhibits a modest advantage in capturing the residual dependence structure. 
Figure~\ref{Fig:acf} displays the ACF of exponential residuals for SE-Burr, ACD-Burr, and FI-logACD models. 


\begin{figure}[ht]
	\centering
	\begin{subfigure}{0.32\textwidth}
		\includegraphics[width=\textwidth]{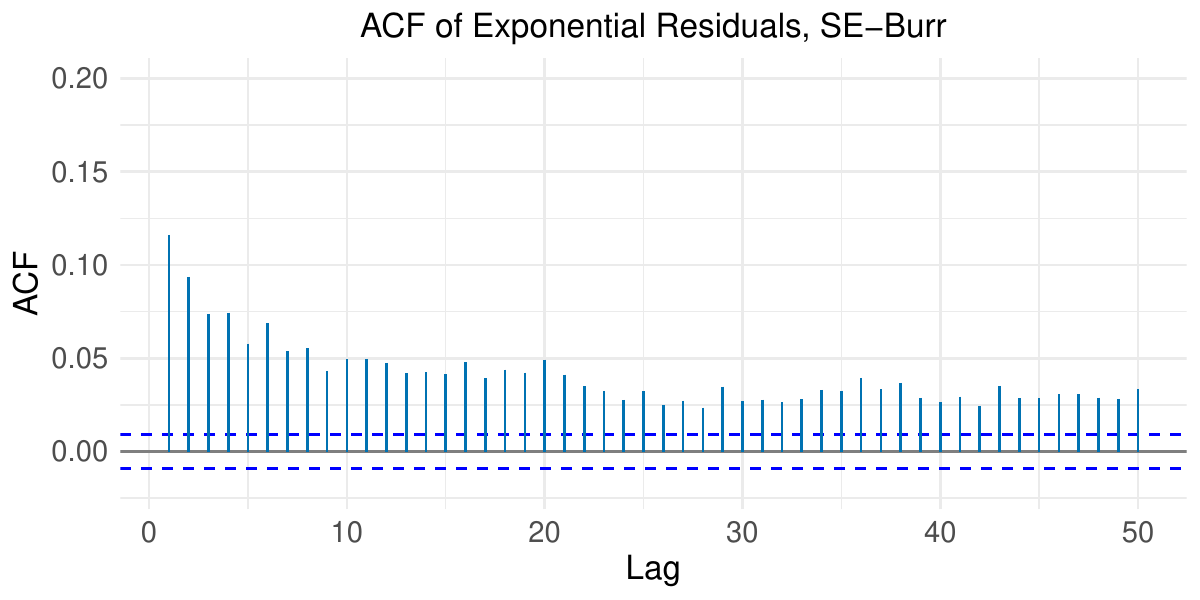}
		\caption{SE--Burr}
	\end{subfigure}
	\hfill
	\begin{subfigure}{0.32\textwidth}
		\includegraphics[width=\textwidth]{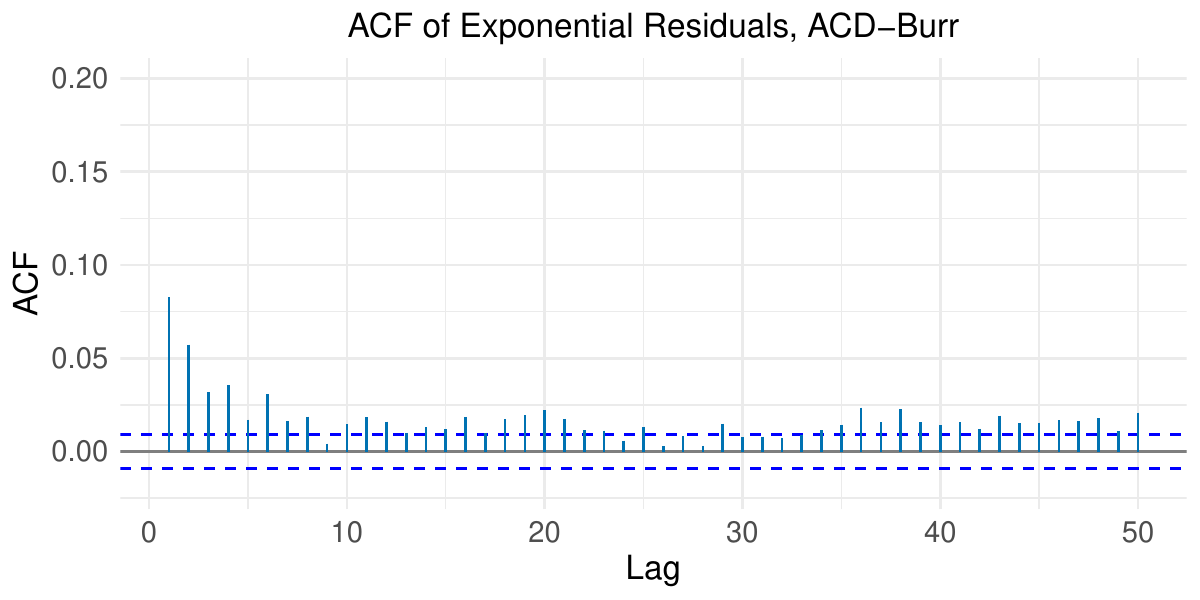}
		\caption{ACD--Burr}
	\end{subfigure}
	\hfill
	\begin{subfigure}{0.32\textwidth}
		\includegraphics[width=\textwidth]{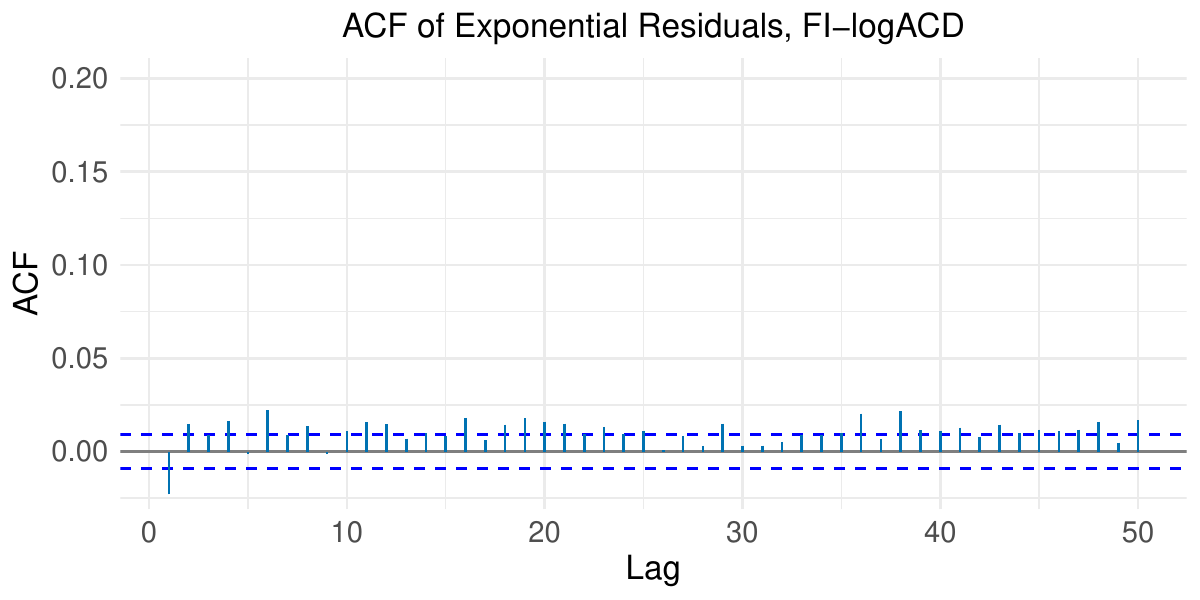}
		\caption{FI-logACD}
	\end{subfigure}
	\caption{Autocorrelation functions of exponential residuals for SE--Burr, 
		ACD--Burr, and FI-logACD models, based on AAPL mid-price data on 
		December 16, 2022.}
	\label{Fig:acf}
\end{figure}

When examining the autocorrelation function of the exponential residuals, the ACD-Burr model displays a faster decay toward 0 as the lag increases, suggesting a quicker loss of temporal dependence compared to SE-Burr. 
Notably, the FI-logACD model essentially eliminates residual autocorrelation across all lags, indicating that a long-memory specification is highly effective at capturing the serial dependence structure of trade durations.
However, this advantage in dependence modeling does not translate into superior out-of-sample forecasting performance (Table~\ref{Table:RMSE}). 
The durations are characterized by heavily right-skewed distributions with thick tails, and under such conditions, 
the choice of residual distribution plays a more decisive role in reducing prediction error than the long-memory dynamics of the conditional mean. 
This explains why the flexible Burr-based specifications, which directly tune the tail behavior of the duration distribution, achieve stronger predictive accuracy despite retaining some residual autocorrelation.

\section{Conclusion}\label{Sect:concl}

This paper applies the self-exciting flexible residual point process to predict price-change durations in the LOB. 
The proposed point process is driven by a self-exciting and exponentially decaying latent process. 
Although this latent process governs the intensity, the choice of residual distribution allows the model to capture a wide variety of distributional patterns observed in practice. 
This flexibility ensures high practical value and suggests potential applications beyond the financial domain.

The process can also be interpreted as a Markov chain. 
Under suitable conditions, this process becomes a positive Harris chain, guaranteeing stationarity. 
This property holds when the residual distribution has a finite mean, providing a solid theoretical foundation for practical applications.
This work establishes this result using the standard Foster--Lyapunov method.

This work further compares the proposed model with several alternatives to assess its predictive performance. 
An empirical study of mid-price duration forecasting found that flexible residual distributions improve predictive accuracy. 
Even compared with the ACD model family, which also accommodates flexible distributions, the proposed model achieved notable improvements in forecasting metrics. 
In addition, the structural design of this model allows for straightforward extensions, enabling a wide range of future applications.

\section*{Acknowledgement}
This work was supported by the National Research Foundation of Korea (NRF) grant funded by the Korea government (MSIT) (No. RS-2026-25469087).

\bibliography{Bib}
\bibliographystyle{chicago}

\appendix

\section{Properties of $\Phi^{-1}$}

First, we examine the mathematical properties of the inverse mapping $\Phi^{-1}(y, x)$ defined in Definition~\ref{Def:FRPP}. 
The function $\Phi(t, x)$ is strictly increasing with respect to $t$ for a fixed $x$;
thus, the inverse $\Phi^{-1}(y, x)$ is well-defined and differentiable in $y$.

\paragraph{Monotonicity in $x$:}  
The inverse function $\Phi^{-1}(y, x)$ is strictly decreasing with respect to $x$,
which can be verified using implicit differentiation:
\[
\frac{\partial}{\partial x} \Phi^{-1}(y, x)
= - \frac{ \frac{1 - \mathrm{e}^{-\beta \Phi^{-1}(y, x)}}{\beta} }{ \mu + (x - \mu + \alpha) \mathrm{e}^{-\beta \Phi^{-1}(y, x)} } < 0.
\]
As $x \to \infty$, the denominator diverges; thus, $\Phi^{-1}(y, x) \to 0$ for any fixed $y > 0$.

\paragraph{Asymptotic behavior in $y$:}  
As $y \to \infty$, the exponential term in $\Phi(t,x)$ vanishes, so that the integral behaves linearly.  
Inverting this relation yields the following approximation
\begin{equation}
	\Phi^{-1}(y, x) \sim \frac{y}{\mu} - \frac{x - \mu + \alpha}{\beta \mu}, 
	\quad \text{as } y \to \infty,
	\label{Eq:Phi_inverse_linear2}
\end{equation}
indicating that the tail behavior of $\varepsilon_n$ is linearly transferred to the interarrival time $\tau_n$.  
Moreover, the linear approximation lies strictly below $\Phi^{-1}(y,x)$, and according to Eq.~\eqref{Eq:phi_FRPP}, the deviation is given by
\[
\Phi^{-1}(y, x) - \frac{y}{\mu} + \frac{x - \mu + \alpha}{\beta \mu}
= \frac{x - \mu + \alpha}{\beta\mu}\, \e^{-\beta t}, 
\quad t = \Phi^{-1}(y,x).
\]

To control this error, fix $\delta_0 > 0$ and define
\begin{equation}
	y_0 = \frac{x - \mu + \alpha}{\beta} 
	+ \frac{\mu}{\beta}\log\!\left( \frac{x - \mu + \alpha}{\mu\beta\delta_0} \right),
	\label{Eq:y_0}
\end{equation}
along with the linear approximation
\[
t_* = \frac{y_0}{\mu} - \frac{x - \mu + \alpha}{\beta \mu}.
\]
For $y \geq y_0$, we have $t_* < t  = \Phi^{-1}(y,x)$, and it follows that
\[
\frac{x - \mu + \alpha}{\beta\mu}\,\e^{-\beta t} 
< \frac{x - \mu + \alpha}{\beta\mu}\,\e^{-\beta t_*}.
\]
Hence, for all $y \geq y_0$, we obtain
\[
\Phi^{-1}(y, x) - \frac{y}{\mu} + \frac{x - \mu + \alpha}{\beta \mu} 
< \frac{x - \mu + \alpha}{\beta\mu} 
\exp\!\left(-\frac{\beta y_0}{\mu} + \frac{x - \mu + \alpha}{\mu} \right)
= \delta_0.
\]
This result shows that the error can be made arbitrarily small by selecting a sufficiently large $y_0$.  
Moreover, by Eq.~\eqref{Eq:y_0}, $y_0$ grows with $x$ at the order $\mathcal{O}(x + \log x)$.

\paragraph{Piecewise linear upper bound.} 
We construct a piecewise linear function that bounds $\Phi^{-1}(y,x)$ from above.  
For $y \geq y_0$, the linear approximation
\[
\frac{y}{\mu} - \frac{x-\mu+\alpha}{\beta\mu} < \Phi^{-1}(y,x)
\]
holds, and the maximal deviation occurs at $y=y_0$, namely
\[
\delta_0(y_0,x) := \Phi^{-1}(y_0,x) - \frac{y_0}{\mu} + \frac{x-\mu+\alpha}{\beta\mu}.
\]
Using this deviation, we define
\begin{equation}
	U_2(y,x) := \frac{y}{\mu} - \frac{x-\mu+\alpha}{\beta\mu} + \delta_0(y_0,x), \label{Eq:U2}
\end{equation}
satisfying $U_2(y,x) \geq \Phi^{-1}(y,x)$ for all $y \geq y_0$.  

For $y \leq y_0$, we instead consider the linear function
\[
U_1(y,x) := \frac{\Phi^{-1}(y_0,x)}{y_0}\, y,
\]
which also satisfies $U_1(y,x) \geq \Phi^{-1}(y,x)$ for $0<y\leq y_0$,
since $\Phi^{-1}(y,x)$ is convex.
This convexity follows from the fact that
\[
\Phi''(t,x) = -\beta (x - \mu + \alpha) \e^{-\beta t} < 0,
\]
implying that $\Phi(\cdot,x)$ is concave in $t$.
By the standard formula for the second derivative of an inverse function,
\[
(\Phi^{-1})''(y,x) = -\frac{\Phi''(t,x)}{(\Phi'(t,x))^3} > 0,
\]
and hence $\Phi^{-1}(\cdot,x)$ is convex in $y$.

Combining these two constructions yields a piecewise linear upper bound for $\Phi^{-1}(y,x)$
\begin{equation}
	U(y,x) :=
	\begin{cases}
		U_1(y,x), & 0<y \leq y_0, \\[6pt]
		U_2(y,x), & y > y_0,
	\end{cases}
	\label{Eq:linear_bound}
\end{equation}
which holds uniformly for all $y>0$.

\section{Mathematical theory and proof for Markov chain}\label{Append:math}

\begin{proposition}
	The conditional density of $\Lambda_{n+1}$ given $\Lambda_n=x$ is
	\[
	f_{\Lambda}(y\mid x)
	= f_\varepsilon \left(\frac{x-y+\alpha}{\beta}-\frac{\mu}{\beta}\log\frac{y-\mu}{x-\mu+\alpha}\right)
	\,\frac{y}{\beta (y-\mu)},\qquad y\in(\mu,x+\alpha].
	\]
\end{proposition}

\begin{proof}
	By construction we have
	\[
	\Lambda_{n+1} = \Psi(\tau_{n+1},x)
	= \mu + (x-\mu+\alpha)\e^{-\beta \tau_{n+1}},
	\]
	where $\tau_{n+1}=\Phi^{-1}(\varepsilon_{n+1},x)$ and
	\[
	\Phi(t,x)=\mu t + (x-\mu+\alpha)\frac{1-\e^{-\beta t}}{\beta}.
	\]
	Using
	\[
	y=\mu+(x-\mu+\alpha)\e^{-\beta\tau},\qquad 
	\varepsilon=\mu\tau+(x-\mu+\alpha)\frac{1-\e^{-\beta\tau}}{\beta},
	\]
	we obtain,
	\[
	\e^{-\beta\tau}=\frac{y-\mu}{x-\mu+\alpha}, \qquad
	\tau=-\frac{1}{\beta}\log\frac{y-\mu}{x-\mu+\alpha}
	\]
	and
	\[
	\varepsilon(y;x)
	=\frac{x-y+\alpha}{\beta}-\frac{\mu}{\beta}\log\frac{y-\mu}{x-\mu+\alpha}.
	\]
	Differentiating with respect to $y$ yields the following
	\[
	\frac{\mathrm{d}\varepsilon(y;x)}{\mathrm{d}y}
	=-\frac{1}{\beta}-\frac{\mu}{\beta(y-\mu)}
	=-\frac{y}{\beta (y-\mu)}.
	\]
	The following results from  applying the change-of-variable formula for densities:
	\[
	f_{\Lambda}(y\mid x)
	= f_\varepsilon\!\left(\varepsilon(y;x)\right)\,
	\left|\frac{\mathrm{d}\varepsilon(y;x)}{\mathrm{d}y}\right|,
	\]
	yielding the stated expression. 
	The support $(\mu, x+\alpha]$ follows from the fact that
	$\tau\ge0$ implies $\Lambda_{n+1}=\mu+(x-\mu+\alpha)\e^{-\beta \tau}\in(\mu,\,x+\alpha]$.
\end{proof}

\begin{proposition}\label{Prop:irreducible}
	The Markov chain $\{\Lambda_n\}$ is Lebesgue-irreducible. 
	For any $x \in \mathcal{X}$ and any measurable set $B \subseteq \mathcal{X}$ with positive Lebesgue measure, there exists an integer $n \geq 1$ such that
	\[
	P^n(x, B) = \mathbb{P}(\Lambda_n \in B \mid \Lambda_0 = x) > 0.
	\]
	Moreover, the Lebesgue measure is maximal for this chain; therefore, the chain is also $\psi$-irreducible.
\end{proposition}

\begin{proof}
	First, the chain can reach any interval in $(\mu, x + \alpha)$ in one step with a positive probability, 
	since $\e^{-\beta \tau}$ has a continuous density supported on $(0, 1)$, regardless of the current state $x$.
	Although the chain cannot reach the point $\{\mu\}$ (a measure-zero set), this does not affect irreducibility.
	
	Next, we show that the chain can reach any interval $[a, b] \subseteq \mathcal{X}$ with positive measure after finitely many steps.
	If \(b > x + \alpha\), the process can gradually increase its state over multiple steps. 
	For example, if \(\e^{-\beta \tau_k} \in (1 - \delta, 1)\) for small $\delta > 0$ and for \(k = 1, \dots, n-1\), then $\Lambda_{n-1}$ can grow sufficiently large with positive probability.
	Given such a large $\Lambda_{n-1}$, the process can move to a target interval $[a, b]$ in one step via appropriate decay, again with positive probability.
	
	Therefore, for any starting point $x$ and any measurable set \(B\) with positive Lebesgue measure, the $n$-step transition probability satisfies \(P^n(x, B) > 0\) for some $n \geq 1$, establishing irreducibility.
\end{proof}

\begin{proposition}\label{Prop:aperiodic}
	The Markov chain $\{\Lambda_n\}$ is aperiodic. 
	The state space \(\mathcal{X}\) cannot be decomposed into disjoint subsets \(D_1, D_2, \dots, D_d\) for some \(d \geq 2\) such that
	\begin{equation}
		P(\Lambda, D_{i+1}) = 1 \quad \forall \Lambda \in D_i, \quad i = 1, 2, \dots, d \ (\text{mod } d). \label{Eq:disjoint}
	\end{equation}
\end{proposition}

\begin{proof}
	Let \(p(x' \mid x)\) denote the transition density, so that for any \(x \in [\mu, \infty)\) and any measurable set \(B \subseteq (\mu, \infty)\),
	\[
	P(x, B) = \int_B p(x' \mid x) \, \D x'.
	\]
	This density has support on the interval \((\mu, x + \alpha)\), which includes an open neighborhood of \(x\). 
	In particular, for any state \(x\), there exists \(\delta > 0\) such that the interval \((x - \delta, x + \delta) \subseteq (\mu, x + \alpha)\) has positive transition probability.
	
	The overlapping nature of the transition kernel's support implies that the chain can return to any small neighborhood of its current state in one step with positive probability. 
	This rules out the existence of a nontrivial cyclic partition of the state space as in \eqref{Eq:disjoint}, and hence the chain is aperiodic.
\end{proof}

\begin{definition}[T-chain]
	Let $a$ be a sampling distribution over the non-negative integers $\mathbb{Z}_+$
	and there exists a substochastic transition kernel $T$ such that
	$$
	K_a(x, B) := \sum_{k=0}^{\infty} a(k) P^k(x, B) \geq T(x, B), \quad x \in \mathcal X, B \in \mathcal B(\mathcal X)
	$$
	where $T(\cdot, B)$ is a lower semicontinuous for any $B \in \mathcal B(\mathcal X)$.
	Then $T$ is called a continuous component of $K_a$.
	If $\{X_n\}$ is a Markov chain for which there exists a sampling distribution $a$ such that $K_a$ possesses a continuous component $T$ with $T(x, \mathcal X) > 0$ for all $x$, then $\{X_n\}$ is called a T-chain. 
\end{definition}

\begin{proposition}
	The Markov chain $\{\Lambda_n\}$ with residual density  $f_\varepsilon$, which is continuous and strictly positive on $(0,\infty)$, is a T-chain.
\end{proposition}

\begin{proof}
	Take the sampling distribution $a$ with $a(1)=1$, so that $K_a = P$. 
	Let $B_0 = (\mu,\mu+\alpha] \subseteq (\mu, x+\alpha]$ for all $x\in\mathcal X=[\mu,\infty)$. 
	Define for $B \subseteq B_0$,
	\[
	T(x,B) := \int_B f_{\Lambda}(y \mid x) \, \mathrm{d}y.
	\]
	Since $f_\varepsilon>0$, the conditional density $f_{\Lambda}(y\mid x)$ is strictly positive on $(\mu,x+\alpha]$, hence $T(x,B_0)>0$ for all $x$. 
	
	By continuity of the conditional density in $x$, we have 
	$f_{\Lambda}(y\mid x_n) \to f_{\Lambda}(y\mid x)$ for each $y\in B$ as $x_n \to x$. 
	Hence, by Fatou’s lemma,
	\[
	\liminf_{n\to\infty} T(x_n,B) 
	= \liminf_{n\to\infty} \int_B f_{\Lambda}(y\mid x_n)\,\mathrm{d}y
	\;\ge\; \int_B f_{\Lambda}(y\mid x)\,\mathrm{d}y = T(x,B),
	\]
	so $x\mapsto T(x,B)$ is lower semicontinuous for every Borel set $B \subseteq B_0$.
	Thus, $T$ is a continuous component of $K_a$ with $T(x,\mathcal X)>0$, and the chain $\{\Lambda_n\}$ is a T-chain.
\end{proof}

In the theory of Markov chains, petite sets play a central role in analyzing stability and recurrence. 
They act as small sets that allow one to control the long-run behavior of the chain through drift and minorization conditions.

\begin{definition}[Petite set]
	Consider a Markov chain with transition kernel $P$ on a general state-space $\mathcal X$.
	A set $A \subset \mathcal{X}$ is said to be $\nu_a$-petite (with respect to $P$), 
	if there exists a non-trivial measure $\nu$ on $\mathcal{X}$ and an integer-valued distribution $a(\cdot)$ on $\mathbb{Z}_+$ such that, for all $x \in A$ and measurable sets $B \subset \mathcal{X}$,
	\[
	\sum_{n=1}^\infty a(n) P^n(x, B) \geq \nu(B).
	\] 
\end{definition}

For a $\psi$-irreducible Markov chain, any $\nu_a$-petite set $A$ is also $\psi_b$-petite for some integer-valued distribution $b(\cdot)$ on $\mathbb{Z}_+$. 
Henceforth, we simply refer to a petite set as one relative to the maximal irreducible measure $\psi$.
In the present work, $\{ \Lambda_n \}$ is $\psi$-irreducible with respect to Lebesgue measure,
so a petite set for $\{ \Lambda_n \}$ is understood relative to Lebesgue measure.

Moreover, for a $\psi$-irreducible T-chain, 
every compact subset of the state space is petite. 

\begin{definition}[Harris recurrent]
	A $\psi$-irreducible Markov chain $X$ is Harris recurrent if, for every measurable set $A \subset \mathcal{X}$ with positive $\psi$-measure, the chain visits $A$ infinitely often, i.e.,
	\[
	\mathbb{P}(\eta_A = \infty \,|\, X_0 = x) = 1, \quad \text{for all } x \in A,
	\]
	where $\eta_A$ denotes the number of visits to the set $A$.
	This concept originates from \cite{harris1956existence}.
\end{definition}

Harris recurrence of $X$ implies that the chain is recurrent, which can be expressed in terms of expected occupation time: 
for any measurable sets $A \subset \mathcal{X}$ with positive measure,
\[
\mathbb{E}[\eta_A \,|\, X_0 = x] = \infty, \quad \text{for all } x \in A.
\]
Equivalently, a $\psi$-irreducible chain $X$ is Harris recurrent if and only if there exists a petite set $C \in \mathcal{B}(\mathcal{X})$ such that
\[
\mathbb{P}(\tau_C < \infty \,|\, X_0 = x) = 1, \quad \text{for all } x \in \mathcal{X},
\]
where $\tau_C$ denotes the first hitting time to $C$.
Furthermore, if for this petite set $C$
\[
\sup_{x \in C} \mathbb{E}[\tau_C \,|\, X_0 = x] < \infty,
\]
then the chain is positive Harris recurrent.
More precisely, the chain admits a unique invariant probability measure 
$\pi$ on $(\mathcal{X}, \mathcal{B}(\mathcal{X}))$, i.e.,
\[
\pi(B) = \int_{\mathcal{X}} P(x,B) \, \pi(dx), 
\quad \forall B \in \mathcal{B}(\mathcal{X}).
\]
Positive Harris recurrence is therefore a central concept in analyzing the stability and long-term behavior of Markov chains, as it ensures both recurrence and the existence of a stationary distribution.

The following theorem, originally introduced by \cite{foster1953stochastic} for discrete-time Markov chains, provides conditions for recurrence and stability.

\begin{theorem}[Foster--Lyapunov criterion]\label{Thm:FL}
	Suppose \( \{X_n\}_{n \in \mathbb{N}}\) is a $\psi$-irreducible Markov chain on a general state-space \(\mathcal{X}\), and that $X$ satisfies a drift condition toward a petite set. 
	More precisely, if there exists a measurable function \(V: \mathcal X \to [0, \infty)\), a petite set \(C \subset \mathcal{X}\), and finite constants \(\delta > 0\) and \(b < \infty\) such that
	\begin{align}
		\mathbb{E}[V(X_{n+1}) \mid X_{n} = x ] - V(x) \leq - \delta + b \mathbbm{1}_{C}(x), \label{Eq:Foster} 
	\end{align}
	then the Markov chain is positive Harris recurrent.
\end{theorem}

To establish positive Harris recurrence, we consider the Lyapunov function $V(x) = x$ 
and define the compact set 
\[
C_K = \{ x \in \mathcal{X} : \mu \leq x \leq K \}.
\] 
Since the chain is a T-chain, $C_K$ is petite.

\begin{lemma}\label{lemma:pHr}
	If there exists a constant \( K \) such that for all \( x > K \),
	\begin{equation}
		\mathbb{E}[\tau_n \mid \Lambda_{n-1} = x] < \frac{\beta \mu_\varepsilon - \alpha}{\beta \mu}, \label{Eq:tau_condition}
	\end{equation}
	then the Markov chain \( \{ \Lambda_n \} \) is positive Harris recurrent.
\end{lemma}

\begin{proof}
	From Eq.~\eqref{Eq:lambda_linear}, we have
	\[
	\mathbb{E}[\Lambda_n \mid \Lambda_{n-1} = x]
	= x + \alpha - \beta \mu_\varepsilon + \beta \mu \, \mathbb{E}[\tau_n \mid \Lambda_{n-1} = x].
	\]
	Under condition~\eqref{Eq:tau_condition} and since $\alpha < \beta \mu_{\varepsilon}$, there exists \( \delta > 0 \) such that
	\begin{equation}
		\mathbb{E}[\Lambda_n \mid \Lambda_{n-1} = x] \leq x - \delta, \label{Eq:nd}
	\end{equation}
	for all \( x > K \). 
	This establishes the negative drift condition required by Eq.~\eqref{Eq:Foster} outside the petite set.
	
	For \( x \leq K \), the drift may be positive or negative, but the state remains within the petite set \( C_K \), ensuring boundedness. 
	Moreover, since \( \tau_n \geq 0 \), we have \( 0 < \e^{-\beta \tau_n} \leq 1 \). From Eq.~\eqref{Eq:lambda}, it follows that
	\[
	\mathbb{E}[\Lambda_n \mid \Lambda_{n-1} = x] \leq K + \alpha,
	\]
	which ensures boundedness within  \( C_K \).
	
	Therefore, both conditions of the Foster--Lyapunov criterion are satisfied, and the chain \(\{\Lambda_n\}\) is positive Harris recurrent.
\end{proof}

\begin{theorem}\label{Thm:positive_recurrent}
	The Markov chain $\{\Lambda_n \}$ is positive Harris recurrent.	
\end{theorem}

\begin{proof}
	Using the piecewise linear upper bound of $\Phi^{-1}$ in Eq.~\eqref{Eq:linear_bound}, we have
	\begin{align*}
		\mathbb{E}[\tau_n \mid \Lambda_{n-1} = x] \leq \mathbb{E}[U_1(\varepsilon, x)  \mathbbm{1}_{\varepsilon \leq y_0} ] + \mathbb{E}[U_2(\varepsilon, x)  \mathbbm{1}_{\varepsilon > y_0} ]
	\end{align*}
	where $y_0$ is defined as in Eq~\eqref{Eq:y_0}.
	Both contributions from $U_1$ and $U_2$ vanish as $x \to \infty$. 
	
	For the first term,
	$$
	\mathbb{E}[U_1(\varepsilon, x)  \mathbbm{1}_{\varepsilon \leq y_0} ] = \frac{\Phi^{-1}(y_0, x)}{y_0} \E \left[   \varepsilon  \mathbbm{1}_{\varepsilon \leq y_0}  \right] \leq \frac{\Phi^{-1}(y_0, x)}{y_0}  \E[\varepsilon] = \frac{1}{y_0} \left(\frac{y_0}{\mu} - \frac{x - \mu + \alpha}{\mu\beta} + \delta_0\right) \E[\varepsilon].
	$$
	By the definition of $y_0$,
	$$
	\mathbb{E}[U_1(\varepsilon, x)  \mathbbm{1}_{\varepsilon \leq y_0} ] \leq \frac{1}{y_0} \left( \log \left( \frac{x - \mu + \alpha}{\mu \beta \delta_0}  \right) + \delta_0 \right) \to 0
	$$
	as $x \to \infty$ since $y_0 = \mathcal O (x + \log x).$ 
	
	For the second term, by Eq.~\eqref{Eq:U2},
	$$
	\mathbb{E}[U_2(\varepsilon, x)  \mathbbm{1}_{\varepsilon > y_0} ] = \frac{1}{\mu}\mathbb{E}\left[\left(\varepsilon  - \frac{x - \mu + \alpha}{\beta} \right) \mathbbm{1}_{\varepsilon > y_0}  \right] +  \delta_0 \mathbb P (\varepsilon > y_0).$$
	As $x \to \infty$, the threshold $y_0 = y_0(x)$ diverges to infinity. 
	By the dominated convergence theorem, since $\varepsilon \mathbbm{1}_{\varepsilon > y_0} \le \varepsilon$ and $\mathbb{E}[\varepsilon] < \infty$, we have
	\[
	\mathbb{E}[\varepsilon \mathbbm{1}_{\varepsilon > y_0}] \to 0 \quad \text{as } y_0 \to \infty,
	\]
	and hence
	$$ \mathbb{E}\left[\left(\varepsilon  - \frac{x - \mu + \alpha}{\beta} \right) \mathbbm{1}_{\varepsilon > y_0}  \right] \leq \mathbb{E}[\varepsilon \mathbbm{1}_{\varepsilon > y_0}] \to 0 \quad \text{as } x \to \infty$$
	and
	$$ \mathbb P (\varepsilon > y_0) \to 0 \quad \text{as } x \to \infty$$
	and
	$$ \mathbb{E}[U_2(\varepsilon, x)  \mathbbm{1}_{\varepsilon > y_0} ] \to 0 \quad \text{as } x \to \infty.$$
	Therefore, the entire expectation \( \mathbb{E}[\tau_n \mid \Lambda_{n-1} = x] \to 0 \) as \( x \to \infty \).
	
	Thus, there exists \( K \) such that for all \( x > K \),
	Eq.~\eqref{Eq:tau_condition} holds.
	By Lemma~\ref{lemma:pHr}, the Markov chain \( \{\Lambda_n\} \) is positive Harris recurrent.
\end{proof}

Under the condition $\alpha < \beta \mu_\varepsilon$, 
the process $\{\Lambda_n\}$ is a Lebesgue-irreducible, aperiodic, and positive Harris recurrent Markov chain. 
Consequently, it admits a unique stationary distribution $\pi$.
For any initial distribution $\nu$, the law $\nu P^n$ converges to $\pi$ in total variation as $n \to \infty$:
\begin{equation*}
	\|\nu P^n - \pi\| = \sup_{A \in \mathcal B(\mathcal X)} \big| \nu P^n(A) - \pi(A) \big| \;\to\; 0, \quad n \to \infty,
\end{equation*}
where
\[
\nu P^n(A) =  \int_{\mathcal X} P^n(x,A)\,\nu(\mathrm{d} x).
\]

In particular, this implies that the process $\{\Lambda_n\}$ is ergodic: 
regardless of the initial distribution $\nu$, its law converges to the stationary distribution $\pi$ in total variation.
If $\Lambda_0$ is initialized according to the stationary distribution, then $\{\Lambda_n\}$ is strictly stationary.
These results are well established in the theory of Markov chains; see, e.g., Meyn and Tweedie~\cite{meyn2009markov} for formal statements and proofs.

The sequence \( \{\tau_n\} \) is also strictly stationary, since \( \tau_n = \Phi^{-1}(\varepsilon_n; \Lambda_{n-1}) \) depends on two stationary sequences.
Under the stability condition \( \alpha < \beta \mu_\varepsilon \),
\( \mathbb{E}[\Lambda_n] \) is finite.
From Eq.~\eqref{Eq:lambda_linear},
\[
\Lambda_n = \Lambda_{n-1} + \alpha - \beta \varepsilon_n + \beta \mu \tau_n.
\]
Taking expectations under stationarity (\(\mathbb{E}[\Lambda_n] = \mathbb{E}[\Lambda_{n-1}] = \mathbb{E}[\Lambda]\)) gives
\[
0 = \alpha - \beta \mu_\varepsilon + \beta \mu \, \mathbb{E}[\tau_n],
\]
which leads to
\[
\mathbb{E}[\tau_n] = \frac{\beta \mu_\varepsilon - \alpha}{\beta \mu}.
\]
\end{document}